\documentclass[11pt]{article}

\usepackage{amsmath,amssymb,amsthm,amsfonts,latexsym,bbm,xspace,graphicx,float,mathtools}
\usepackage[backref,colorlinks,citecolor=blue,bookmarks=true]{hyperref}
\usepackage[letterpaper,margin=1in]{geometry}
\usepackage{color}
\usepackage{algorithm}
\usepackage{algorithmic}

\usepackage{tweaklist}
\renewcommand{\enumhook}{\setlength{\topsep}{2pt}%
  \setlength{\itemsep}{-2pt}}
\renewcommand{\itemhook}{\setlength{\topsep}{2pt}%
  \setlength{\itemsep}{-2pt}}

\newtheorem{theorem}{Theorem}
\newtheorem{informal}{Informal Theorem}
\newtheorem{corollary}{Corollary}
\newtheorem{lemma}{Lemma}

\newtheorem{proposition}{Proposition}

\newtheorem{definition}{Definition}

\newtheorem{observation}{Observation}

\newtheorem{remark}{Remark}

\newenvironment{prevproof}[2]{\noindent {\em {Proof of {#1}~\ref{#2}:}}}{$\Box$\vskip \belowdisplayskip}

\newcommand{\junk}[1]{}

\junk{

}

\newcommand{\poly}{\text{poly}}

\newcommand{\mattnote}[1]{{\color{black}{#1}}}

\newcommand{\notshow}[1]{{}}

\def \obj {{\cal O}}

\DeclareMathOperator{\argmax}{argmax}

\definecolor{MyGray}{rgb}{0.8,0.8,0.8}

\begin{document}
\title{{Understanding Incentives: Mechanism Design becomes Algorithm Design}}
\author {Yang Cai\thanks{Supported by NSF Award CCF-0953960 (CAREER) and CCF-1101491.}\\
EECS, MIT \\
\tt{ycai@csail.mit.edu}
\and
Constantinos Daskalakis\thanks{Supported by a Sloan Foundation Fellowship, a Microsoft Research Faculty Fellowship and NSF Award CCF-0953960 (CAREER) and CCF-1101491.}\\
EECS, MIT \\
\tt{costis@mit.edu}
\and
S. Matthew Weinberg\thanks{Supported by a NSF Graduate Research Fellowship and NSF award CCF-1101491.}\\
EECS, MIT\\
\tt{smw79@mit.edu}
}
\addtocounter{page}{-1}
\maketitle
\begin{abstract}
We provide a computationally efficient black-box reduction from mechanism design to algorithm design in very general settings. Specifically, we give an approximation-preserving reduction from truthfully maximizing \emph{any} objective under \emph{arbitrary} feasibility constraints with \emph{arbitrary} bidder types to (not necessarily truthfully) maximizing the same objective plus virtual welfare (under the same feasibility constraints). Our reduction is based on a fundamentally new approach: we describe a mechanism's behavior indirectly only in terms of the expected value it awards bidders for certain behavior, and never directly access the allocation rule at all. 

Applying our new approach to revenue, we exhibit settings where our reduction holds \emph{both ways}. That is, we also provide an {approximation-sensitive} reduction from (non-truthfully) maximizing virtual welfare to (truthfully) maximizing revenue, and therefore the two problems are computationally equivalent. With this equivalence in hand, we show that both problems are NP-hard to approximate within any polynomial factor, even for a single monotone submodular bidder.

We further demonstrate the applicability of our reduction by providing a truthful mechanism maximizing fractional max-min fairness. This is the first instance of a truthful mechanism that optimizes a non-linear objective.
\end{abstract}
\thispagestyle{empty}

\newpage

\section{Introduction} \label{sec:intro}

{\em Mechanism design} is the problem of optimizing an objective subject to ``rational inputs.'' The difference to {\em algorithm design} is that the inputs to the objective are not known, but are owned by rational agents who need to be provided incentives in order to share enough information about their inputs  such that the desired objective can be optimized. The question that arises is how much this added complexity degrades our ability to optimize objectives, namely

\bigskip ~~~~\begin{minipage}[h]{14.5cm}
{\em How much more computationally difficult is mechanism design for a certain objective compared to algorithm design for the same objective?}
\end{minipage}

\bigskip  This question has been at the forefront of algorithmic mechanism design, starting already with the seminal work of Nisan and Ronen~\cite{NisanR99}. In a non-Bayesian setting, i.e. when no prior distributional information is known about the inputs, we now have strong separation results between  algorithm and mechanism design. Indeed, a sequence of recent breakthroughs~\cite{PapadimitriouSS08,BuchfuhrerDFKMPSSU10,Dobzinski11,DobzinskiV12} has culminated in combinatorial auction settings where welfare can be  optimized computationally efficiently to within a constant factor for ``honest inputs,'' but it cannot be computationally efficiently optimized to within a polynomial factor for ``rational inputs,'' subject to well-believed complexity theoretic assumptions. Besides, the work of Nisan and Ronen studied the problem of minimizing makespan on unrelated machines, which can be well-approximated for honest machines, but whose approximability for rational machines still remains unknown.

 In a Bayesian world, where every input is drawn from some known distribution, algorithm and mechanism design appear more tightly connected. Indeed, a sequence of surprising works~\cite{HartlineL10,HartlineKM11,BeiH11} have established that mechanism design for welfare optimization in an arbitrary environment\footnote{An {\em environment} constrains the feasible outcomes of the mechanism as well as the allowable bidder {\em types}, or {\em valuations}. The latter map outcomes to value units.}  can be computationally efficiently reduced to algorithm design in the same environment, in an approximation-preserving way. A similar reduction has been recently discovered for the revenue objective~\cite{CaiDW12b,CaiDW13} in the case of {\em additive bidders}.\footnote{A bidder is {\em additive} if her value for a bundle of items is just the sum of her values for each item in the bundle.} Here, mechanism design for revenue optimization in an arbitrary additive environment (computationally efficiently) reduces to algorithm design for {\em virtual welfare optimization} in the same environment, in an approximation-preserving manner. The natural question is whether such  mechanism- to algorithm-design reduction is achievable for general bidder types (i.e. beyond additive) and  general objectives (i.e. beyond revenue and welfare). This is what we achieve in this paper.
 
 \begin{informal} \label{inf thm 1} There is a generic, computationally efficient, approximation-preserving reduction from mechanism design for an arbitrary {\em concave} objective $\cal O$, under arbitrary feasibility constraints and arbitrary allowable bidder types, to algorithm design, under the same feasibility constraints and allowable bidder types, and objective:
\begin{itemize}
 \item $\cal O$ plus virtual welfare, if $\cal O$ is an \mattnote{allocation-only objective} \mattnote{(i.e. $\cal O$ depends only on the allocation chosen and not on payments made).};
 \item $\cal O$ plus virtual welfare plus virtual revenue, if $\cal O$ is a general objective \mattnote{(i.e. $\cal O$ may depend on the allocation chosen as well as payments made)};
 \item virtual welfare, if $\cal O$ is the revenue objective.
 \end{itemize}
 \end{informal}

A formal statement of our result is provided as Theorem~\ref{thm:objective} in Section~\ref{sec:MDMDPtoSADPgeneral}. Specifically, we provide a Turing reduction from the {\em Multi-Dimensional Mechanism Design Problem} (MDMDP) to the {\em Solve Any-Differences Problem} (SADP). MDMDP and SADP are formally defined in Section~\ref{sec:generalprelim}. They are both parameterized by a set ${\cal F}$, specifying feasibility constraints on outcomes,\footnote{These could encode, e.g., matching constraints of a collection of items to the bidders, or possible locations to build a public project, etc.} a set ${\cal V}$ of functions, specifying allowable types of bidders,\footnote{A {\em type} $t$ of a bidder is a function mapping $\cal F$ to the reals, specifying how much the bidder values every outcome in ${\cal F}$. If a set ${\cal V}$ of functions parameterizes one of our problems, then all bidders are restricted to have types in ${\cal V}$. E.g., ${\cal F}$ may be $\mathbb{R}^\ell$ and ${\cal V}$ may contain all \mattnote{additive} functions over ${\cal F}$. } and an objective function ${\cal O}$, mapping a profile $\vec{t} \in {\cal V}^m$ of bidder types ($m$ is the number of bidders), a distribution $X \in \Delta({\cal F})$ over feasible outcomes, and a randomized price vector $P$, to the reals.\footnote{E.g. ${\cal O}$ could be revenue (in this case, ${\cal O}(\vec{t}, X, P)= \mathbb{E}[\sum_i P_i]$), or it could be welfare (in this case, ${\cal O}(\vec{t}, X, P)= \sum_i t_i(X)$, where $t_i(X)=\mathbb{E}_{x \leftarrow X}[t(x)]$ is the expected value of type $t_i$ for the distribution over outcomes $X$), or it could be some fairness objective such as ${\cal O}(\vec{t}, X, P)= \min_i t_i(X)$.} In terms of these parameters:
\begin{itemize}
\item MDMDP is the problem of designing a mechanism $M$ that maximizes ${\cal O}$ in expectation over the types $t_1,\ldots,t_m$ of the bidders, given a product distribution over ${\cal V}^m$ for $t_1,\ldots,t_m$, and assuming that the bidders play $M$ truthfully. $M$ is restricted to choose outcomes in ${\cal F}$ with probability $1$, it must be Bayesian Incentive Compatible, and Individually Rational.

\item SADP is input a type vector $t_1,\ldots,t_m \in {\cal V}$, a list of hyper-types $\mathfrak{t}_1,\ldots,\mathfrak{t}_k \in {\cal V}^*$, where ${\cal V}^*$ is the closure of ${\cal V}$ under addition and positive scalar multiplication, and weights $c_0 \in \mathbb{R}_{\ge 0}$ and $c_1,\ldots,c_m \in \mathbb{R}$. The goal is to choose a distribution $X$ over outcomes and a randomized price vector $P$ so that 
$$c_0 {\cal O}(\vec{t}, X, P) + \sum_i c_i \mathbb{E}[P_i] + (\mathfrak{t}_j(X) - \mathfrak{t}_{j+1}(X)).$$
is maximized for at least one $j \in \{1,\ldots,k-1\}$. The first term in the above expression is a scaled version of ${\cal O}$, the second is a ``virtual revenue'' term (where $c_i$ is the virtual currency that bidder $i$  uses), and the last term is a ``virtual welfare'' term (of a pair of adjacent hyper-types the first of which is scaled by $1$ and the other by $-1$). The name ``Solve Any-Differences Problem'' alludes to the freedom of choosing any value of $j$ and then optimizing.
\end{itemize}

Notice that MDMDP is a mechanism design problem, where our task is to optimize objective ${\cal O}$ given distributional information about the bidder types. On the other hand, SADP is an algorithm design problem, where types and hyper-types are perfectly known and the task is to optimize the sum of the same objective ${\cal O}$ plus a virtual revenue and welfare term. In this terminology, Theorem~\ref{thm:objective} (stated above as Informal Theorem~\ref{inf thm 1}) establishes that there is a computationally efficient reduction from $\alpha$-approximating MDMDP to $\alpha$-approximating SADP, for any value $\alpha$ of the approximation, as long as $\cal O$ is concave in $X$ and $P$.\footnote{Formally ${\cal O}$ is concave in $X$ and $P$ if for any $(X_1, P_1)$ and $(X_2,P_2)$ and any $c \in [0,1]$ and $\vec{t}$, $\obj(\vec{t},cX_1 + (1-c)X_2,cP_1 + (1-c)P_2) \geq c\obj(\vec{t},X_1,P_1) + (1-c)\obj(\vec{t},X_2,P_2)$, where $cX_1 + (1-c)X_2$ is the mixture of distributions $X_1$ and $X_2$ over outcomes, with mixing weights $c$ and $1-c$ respectively.} 

It is worth stating a few caveats of our reduction:
\begin{enumerate}
\item First, it is known from~\cite{ChawlaIL12} that it is not possible to have a general reduction from mechanism design to algorithm design with the exact same objective. This motivates the need to include the extra terms of virtual revenue and virtual welfare in the objective of SADP.

\item If $\cal O$ is \mattnote{allocation-only}, i.e. it does not depend on the price vector $P$, then all coefficients $c_1,\ldots,c_m$ can be taken $0$ in the reduction from MDMDP to SADP. Hence,  to $\alpha$-approximate MDMDP it suffices to be able to $\alpha$-optimize $\cal O$ plus virtual welfare. In Sections~\ref{sec:maxmin intro} and~\ref{sec:maxmin}, we discuss {\em fractional max-min fairness} as an example of such an objective, providing optimal mechanisms for it through our reduction.

\item If $\cal O$ is \mattnote{price-only}, i.e. it does not depend on the outcome $X$, then the objective in SADP is separable into a price-dependent component ($\cal O$ plus virtual revenue) and an outcome-dependent component (virtual welfare). Hence, our reduction implies that to $\alpha$-approximate MDMPD it suffices to be able to $\alpha$-optimize each of these components separately.\label{item:price dependent}

\item If $\cal O$ is the revenue objective (this is a special case of~\ref{item:price dependent}), the price-dependent component in SADP is trivial to optimize. In this case, to $\alpha$-approximate MDMPD it suffices to be able to $\alpha$-optimize virtual welfare (i.e. we can take $c_0=c_1=\cdots=c_m=0$ in the SADP instance output by the reduction). See Theorem~\ref{thm:revenue}. Additionally we note the following. 
\begin{enumerate}
\item This special case of our reduction already generalizes the results of~\cite{CaiDW12b} to arbitrary types. Recall that the reduction of~\cite{CaiDW12b} from MDMDP to virtual-welfare optimization could only accommodate additive types.
\item For a special family $\cal V$ of functions, we provide a reduction in the other direction, i.e. from SADP to MDMDP. As a corollary of this reduction we obtain strong inapproximability results for optimal multi-dimensional  mechanism design with submodular bidders. We discuss this in more detail in Section~\ref{subsec: intro extra results for revenue}.
\end{enumerate}

\item Finally, our generic reduction from MDMDP to SADP can take the number $k$ of hyper-types input to SADP to be $2$. We define SADP for general $k$ for flexibility. In particular, general $k$ enables our inapproximability result for optimal mechanism design via a reduction from SADP (general $k$) to MDMDP.

\end{enumerate}

\subsection{Revenue}\label{subsec: intro extra results for revenue}

Our framework described above provides reductions from mechanism design for some arbitrary objective ${\cal O}$ to algorithm design for the same objective ${\cal O}$ plus a virtual revenue and a virtual welfare term. As pointed out earlier in this section, we can't avoid some modification of ${\cal O}$ in the algorithm design problem sitting at the output of a general reduction such as ours, due to the impossibility result of~\cite{ChawlaIL12}. Nevertheless, there could very well be other modified objectives that a general reduction could be reducing to, with better or worse algorithmic properties. The question that arises is this: {\em Could we be hurting ourselves focusing on SADP as an algorithmic vehicle to solve MDMDP?} Our previous work on revenue maximization for additive bidders~\cite{CaiDW12b} exhibits very general ${\cal F}$'s where the answer is ``no,'' motivating our generalization here to non-additive bidders and general objectives. Indeed, we illustrate the reach of our new framework in Section~\ref{sec:maxmin intro} by providing optimal mechanisms for non-linear objectives, an admittedly difficult and under-developed topic in Bayesian mechanism design~\cite{ChawlaIL12,ChawlaHMS13}. 

Here we provide a different type of evidence for the tightness of our approach via reductions going the other way, i.e. from SADP to MDMDP. Recall that MDMDP$({\cal F}, {\cal V}, {\rm Revenue})$ reduces to solving SADP instances, which satisfy $c_0=c_1=\cdots=c_m=0$ and therefore only have a virtual welfare component depending on some $\mathfrak{t}_1,\ldots,\mathfrak{t}_k \in {\cal V}^*$. In Section~\ref{sec:SADPtoMDMDP}, we identify conditions for a collection of functions $\mathfrak{t}_1,\ldots,\mathfrak{t}_k \in {\cal V}^*$ under which SADP reduces to MDMDP, showing that for such instances solving SADP is unavoidable for solving MDMDP. Indeed, our reduction is strong enough that we obtain very strong inapproximability results for revenue optimization, even when there is a single monotone submodular bidder. To the best of our knowledge our result is the first inapproximability result for optimal mechanism design.

\begin{informal} \label{inf thm inapprox} {\em MDMDP($2^{[n]}$, {monotone submodular functions}, {revenue})} cannot be approximated to within any polynomial factor in polynomial time, if we are given value or demand oracle access to the sub-modular functions in the support of the bidders distributions,\footnote{We explain the difference between value and demand oracle access in Section~\ref{sec:SADPtoMDMDP}.} even if there is only one bidder. The same is true if we are given explicit access to these functions (as Turing machines) unless $NP \subseteq RP$.
 \end{informal}

\subsection{Fractional Max-Min Fairness} \label{sec:maxmin intro}

Certainly revenue and welfare are the most widely studied objectives in mechanism design. Nevertheless, resource allocation often requires optimizing non-linear objectives such as the fairness of an allocation, or the makespan of some scheduling of jobs to machines. Already the seminal paper of Nisan and Ronen studies minimizing makespan when scheduling jobs to selfish machines, in a non-Bayesian setting. Following this work, a lot of algorithmic mechanism design research has focused on non-linear objectives in non-Bayesian settings (see, e.g.,~\cite{christodoulou2007lower,ashlagi2012optimal} and their references), but positive results have been scarce. More recently, research has studied non-linear objectives in Bayesian settings~\cite{ChawlaIL12,ChawlaHMS13}. While~\cite{ChawlaIL12} provide impossibility results, the results of~\cite{ChawlaHMS13} give hope that non-linear objectives might be better behaved in Bayesian settings. In part, this is our motivation for providing an algorithmic framework for general objectives in this work. 

As a concrete example of the reach of our techniques, we provide optimal mechanisms for a (non-linear) max-min fairness objective in Section~\ref{sec:maxmin}. The setting we solve is this: There are $n$ items that can be allocated to $m$ additive bidders, subject to some constraints ${\cal F}$. ${\cal F}$ could be matching constraints, matroid constraints, downwards-closed constraints, or more general constraints. Now, given a distribution $X$ over allocations in ${\cal F}$, how fair is it? E.g., if there is one item and two bidders with value $1$ for the item, what is the fairness of a randomized allocation that gives the item to each bidder with probability $1 \over 2$? Should it be $0$, because with probability $1$, exactly one bidder gets value $0$ from the allocation? Or, should it be $1/2$ because each bidder gets an expected value of $1/2$? Clearly, both are reasonable objectives, but we study the latter. Namely, we define the {\em fractional max-min fairness} objective as: $${\cal O}(\vec{t},X) = \min_i t_i (X).$$
We obtain the following result, which is stated formally as Corollary~\ref{cor:maxmin} in Section~\ref{sec:maxmin}.
 \begin{informal} \label{inf thm max-min} Let $G$ be a polynomial-time $\alpha$-approximation algorithm for
 $$\text{{\em \bf Max-Weight($\mathcal{F}$)}: Given weights $(w_{ij})_{ij}$, find $S \in {\cal F}$ maximizing $\sum_{(i,j) \in S}w_{ij}$.}$$
 With black-box access to $G$, we can $\alpha$-approximate MDMDP$({\cal F}, \text{additive functions}, {\cal O})$ in polynomial time. For instance, if ${\cal F}$ are matching constraints, matroid constraints, or the intersection of two matroids, we can optimally solve MDMDP$({\cal F}, \text{additive functions}, {\cal O})$ in polynomial-time.
 \end{informal}

\subsection{Related Work}
\paragraph{Revenue Maximization.} There has been much work in recent years on revenue maximization in multi-dimensional settings~\cite{Alaei11,AlaeiFHHM12,BriestCKW10,CaiD11,CaiDW12,CaiH13,ChawlaHK07,ChawlaHMS10, DobzinskiFK11, HartN12}. Our approach is most similar to that of~\cite{CaiDW12b,CaiDW13}, which was recently extended in~\cite{BGM13}. These works solved the revenue maximization problem for \emph{additive} bidders via a black-box reduction to welfare maximization. In~\cite{BGM13}, numerous extensions are shown that accommodate risk-averse buyers, ex-post budget constraints, and more. But both approaches are inherently limited to revenue maximization and additive bidders. Even just within the framework of revenue maximization, our work breaks through a major barrier, as every single previous result studies only additive bidders.

\paragraph{Hardness of Revenue Maximization.} Three different types of results regarding the computational hardness of revenue maximization are known. It is shown in~\cite{BriestK07} that (under standard complexity theoretic assumptions) no efficient algorithm can find a deterministic mechanism whose revenue is within any polynomial factor of the optimal (for deterministic mechanisms), even for very simple single bidder settings. However, the optimal randomized mechanism in those same settings can be found in polynomial time~\cite{BriestCKW10}. Hardness results for randomized mechanisms are comparatively scarce.

Very recently, a new type of hardness was shown in~\cite{DaskalakisDT13}. There, they show that it is \#P-hard to find (any description of) an optimal randomized mechanism  even in very simple single bidder settings. Specifically, the problem they study is of a single additive bidder whose value for each of $n$ items is drawn independently from a distribution of support 2. The natural description complexity of this problem is $O(n)$ (just list the values for each item and their probabilities), but they show that the optimal randomized mechanism cannot be found or even executed in time $\poly(n)$ (unless ZPP = \#P). This is a completely different type of hardness than what is shown in this paper. Specifically, we show that certain instances are hard to solve even when the support of the input distribution is small (whereas it is $2^n$ in the hard examples of~\cite{DaskalakisDT13}), but the instances are necessarily more involved (we use submodular bidders), as the optimal randomized mechanism can be found in time polynomial in the support of the input distribution for additive bidders~\cite{BriestCKW10}.

The existing result that is most similar to ours appears in~\cite{DobzinskiFK11}. There, they show that it is NP-hard to maximize revenue exactly when there is a single bidder whose value for subsets of $n$ items is an OXS function.\footnote{OXS functions are a subclass of submodular functions} Our approaches are even somewhat similar: we both aim to understand the necessary structure on a type space in order for the optimal revenue to satisfy a simple formula. The big difference between their result and ours is that their results are inherently limited to settings with a single bidder who has \emph{two} possible types. While this suffices to show hardness of exact maximization, there is no hope of extending this to get hardness of approximation.\footnote{The seller always has the option of completely ignoring one type and charging the other their maximum value for their favorite set. This mechanism achieves a $\frac{1}{2}$-approximation in every setting.} Our stronger results are enabled by a deeper understanding of the optimal revenue for single bidder settings with many possible types, which is significantly more involved than the special case of two types. 

\paragraph{General Objectives.} Following the seminal paper of Nisan and Ronen, much work in algorithmic mechanism design has been devoted to maximizing non-linear objectives in a truthful manner. Recently, more attention has been given to Bayesian settings, as there are numerous strong hardness results in non-Bayesian settings. Still, it is shown in~\cite{ChawlaIL12} that no polynomial-time black-box reduction from truthfully maximizing a non-linear objective to non-truthfully maximizing the same non-linear objective exists without losing a polynomial factor in the approximation ratio, even in Bayesian settings. Even more recently, a non-black box approach was developed in~\cite{ChawlaHMS13} to minimize makespan in certain Bayesian settings. Our black-box approach sidesteps the hardness result of~\cite{ChawlaIL12} by reducing the problem of truthfully maximizing an objective to non-truthfully maximizing a modified objective.

\subsection{Paper Structure} To make our framework easier to understand, we separate the paper as follows. In Sections~\ref{sec:prelim} through~\ref{sec:SADPtoMDMDP}, we   provide the necessary details of our framework to show how it applies to revenue maximization. Then, in Section~\ref{sec:general}, we display the full generality of our approach, exemplifying how it applies to the fractional max-min fairness objective in Section~\ref{sec:maxmin}. To ease notation we initially define restricted versions of the MDMDP and SADP problems in Section~\ref{sec:prelim} as they apply to revenue, using these restricted definitions through Section~\ref{sec:SADPtoMDMDP}. Then, in Section~\ref{sec:general}, we expand these definitions to accommodate general objectives.

\section{Preliminaries}\label{sec:prelim}
\paragraph{Mechanism Design Setting.} The mechanism designer has a set of feasible outcomes $\mathcal{F}$ to choose from, which depending on the application could be feasible allocations of items to bidders, locations to build a public project, etc. Each bidder participating in the mechanism may have several possible {\em types}. A bidder's {type} consists of a value for each possible outcome in $\mathcal{F}$. Specifically, a bidder's type is a function $t$ mapping $\mathcal{F}$ to $\mathbb{R}_+$. $T_i$ denotes the set of all possible types of bidder $i$, which we assume to be finite. The designer has a prior distribution $\mathcal{D}_i$ over $T_i$ for bidder $i$'s type. Bidders are {\em quasi-linear} and {\em risk-neutral}. That is, the utility of a bidder of type $t$ for a randomized outcome (distribution over outcomes) $X \in \Delta(\mathcal{F})$, when he is charged (a possibly random price with expectation) $p$, is $\mathbb{E}_{x \leftarrow X}[t(x)] - p$. Therefore, we may extend $t$ to take as input distributions over outcomes as well, with $t(X) = \mathbb{E}_{x \leftarrow X}[t(x)]$. A {\em type profile} $\vec{t}=(t_1,\ldots,t_m)$ is a collection of types for each bidder. We assume that the types of the bidders are independent so that $\mathcal{D} = \times_i \mathcal{D}_i$ is the designer's prior distribution over the complete type profile. 

\paragraph{Mechanisms.} A (direct) mechanism consists of two functions, a (possibly randomized) allocation rule and a (possibly randomized) price rule, and we allow these rules to be correlated. The allocation rule takes as input a type profile $\vec{t}$ and (possibly randomly) outputs an allocation $A(\vec{t}) \in \mathcal{F}$. The price rule takes as input a profile $\vec{t}$ and (possibly randomly) outputs a price vector $P(\vec{t})$. When the bid profile $\vec{t}$ is reported to the mechanism $M = (A,P)$, the (possibly random) allocation $A(\vec{t})$ is selected and bidder $i$ is charged the (possibly random) price $P_i(\vec{t})$. We will sometimes discuss the \emph{interim allocation rule} of a mechanism, which is a function that takes as input a bidder $i$ and a type $t_i \in T_i$ and outputs the distribution of allocations that bidder $i$ sees when reporting type $t_i$ over the randomness of the mechanism and the other bidders' types. Specifically, if the interim allocation rule of $M=(A,P)$ is $X$, then $X_i(t_i)$ is a distribution satisfying $$\Pr[x\leftarrow X_i(t_i) ] = \mathbb{E}_{\vec{t}_{-i} \leftarrow \mathcal{D}_{-i}}\left[\Pr[A(t_i;\vec{t}_{-i}) = x\ |\ \vec{t}_{-i}]\right],$$
where $t_{-i}$ is the vector of types of all bidders but bidder $i$ in $\vec{t}$, and ${\cal D}_{-i}$ is the distribution of $t_{-i}$. Sometimes we write $\vec{t}_{-i}$ instead of $t_{-i}$ to emphasize that it's a vector of types.

\smallskip A mechanism is said to be {\em Bayesian Incentive Compatible (BIC)} if it is in every bidder's best interest to report truthfully their type, conditioned on the fact that the other bidders report truthfully their type. A mechanism is said to be {\em Individually Rational (IR)} if it is in every bidder's best interest to participate in the mechanism, no matter their type. These definitions are given formally in Section~\ref{app:prelim}.

\paragraph{Goal of the designer.} In Section~\ref{sec:revenue} we present our mechanism- to algorithm-design reduction for the revenue objective. The problem we reduce from is designing a BIC, IR mechanism that maximizes expected revenue, when encountering a bidder profile sampled from some given distribution $\mathcal{D}$. Our reduction is described in terms of the problems MDMDP and SADP defined next. In Section~\ref{sec:general} we generalize our reduction to general objectives and accordingly generalize both problems to accommodate general objectives. But our approach is easier to understand for the revenue objective, so we give that first.

\paragraph{Formal Problem Statements.} We present black-box reductions between two problems: the Multi-Dimensional Mechanism Design Problem (MDMDP) and the Solve-Any Differences Problem (SADP). MDMDP is a well-studied mechanism design problem~\cite{CaiDW12,CaiDW12b,CaiDW13}. SADP is a new algorithmic problem that we show has strong connections to MDMDP. In order to discuss our reductions appropriately, we will parameterize the problems by two parameters ${\cal F}$ and ${\cal V}$. Parameter $\mathcal{F}$ denotes the feasibility constraints of the setting; e.g., $\mathcal{F}$ might be ``each item is awarded to at most one bidder'' or ``a bridge may be built in location $A$ or $B$'', etc.  Parameter $\mathcal{V}$ denotes the allowable valuation functions, mapping ${\cal F}$ to the reals; e.g., if ${\cal F}=\mathbb{R}^\ell$, then $\mathcal{V}$ may be ``all additive functions over ${\cal F}$'' or ``all submodular functions'', etc. 
Informally, MDMDP asks for a BIC, IR mechanism that maximizes expected revenue for certain feasibility constraints ${\cal F}$ and a restricted class of valuation functions ${\cal V}$. SADP asks for an element in ${\cal F}$ maximizing the difference of two functions in ${\cal V}$, but the algorithm is allowed to choose any two adjacent functions in an ordered list of size $k$. Throughout the paper will use $\mathcal{V}^*$ to denote the closure of $\mathcal{V}$ under addition and positive scalar multiplication, {and the term ``$\alpha$-approximation'' ($\alpha \leq 1$) to denote a (possibly randomized) algorithm whose expected value for the desired objective is an $\alpha$-fraction of the optimal.}

\smallskip \textbf{MDMDP($\mathcal{F}$, $\mathcal{V}$):} {\sc Input:} For each bidder $i \in [m]$, a finite set of types $T_i \subseteq {\cal V}$ and a distribution ${\cal D}_i$ over $T_i$. {\sc Goal:} Find a feasible (outputs an outcome in $\mathcal{F}$ with probability $1$) BIC, IR mechanism~$M$, that maximizes expected revenue, when $n$ bidders with types sampled from $\mathcal{D} = \times_i \mathcal{D}_i$ play $M$ truthfully (with respect to all feasible, BIC, IR mechanisms). $M$ is said to be an $\alpha$-approximation to MDMDP if its expected revenue is at least a $\alpha$-fraction of the optimal obtainable expected revenue.

\smallskip \textbf{SADP($\mathcal{F}$, $\mathcal{V}$):} Given as input functions $f_{j} \in \mathcal{V}^{*}$ ($1 \leq j \leq k$), find a feasible outcome $X \in \mathcal{F}$ such that there exists an index $j^* \in [k-1]$ such that:

$$f_{j^*}(X) - f_{j^*+1}(X) = \max_{X' \in \mathcal{F}} \{f_{j^*}(X') - f_{j^*+1}(X')\}.$$

$X$ is said to be an $\alpha$-approximation to SADP if there exists an index $j^* \in [k-1]$ such that:

$$f_{j^*}(X) - f_{j^*+1}(X) \geq \alpha \max_{X' \in \mathcal{F}} \{f_{j^*}(X') - f_{j^*+1}(X')\}.$$

%
%
%
%

\paragraph{Representation Questions.} Notice that both MDMDP and SADP are parameterized by ${\cal F}$ and ${\cal V}$. As we aim to leave these sets unrestricted, we assume that their elements are represented in a computationally meaningful way. That is, we assume that elements of $\mathcal{F}$ can be indexed using $O(\log |{\cal F}|)$ bits {and are input to functions that evaluate them via this representation. We assume elements $f \in \mathcal{V}$ are input either via a turing machine that evaluates $f$ (and the size of this turing machine counts towards the size of the input), or as a black box. Moreover, all of our reductions apply whether or not the input functions are given explicitly or as a black box \footnote{{When we claim that we can solve problem $\mathcal{P}_1$ given black-box access to a solution to problem $\mathcal{P}_2$, we mean that the functions input to problem $\mathcal{P}_1$ may be given either explicitly or as a black box, and that they are input in the same form to $\mathcal{P}_2$.}}.} Finally, whenever we evaluate the running time of an  algorithm for either MDMDP or SADP, or of a reduction from one problem to the other, we count the time spent in an oracle call to functions input to these problems as one. Similarly, whenever we show a computational hardness result for either MDMDP or SADP, the time spent in one oracle call is considered as one.

%


\paragraph{Linear Programming.} Our results require the ability to solve linear programs with separation oracles as well as ``weird'' separation oracles, a concept recently introduced in~\cite{CaiDW13}. Throughout the paper we will use the notation $\alpha P$ to denote the polytope $P$ shrunk by a factor of $\alpha \leq 1$. That is, $\alpha P = \{\alpha \vec{x}| \vec{x} \in P\}$. We make use of the following Theorem from~\cite{CaiDW13}, as well as other theorems regarding solving linear programs which are included in Section~\ref{app:linearprogramming}.

\begin{theorem}\label{thm:CaiDW1}(\cite{CaiDW13}) 
Let $P$ be a $d$-dimensional bounded convex polytope containing the origin, and let $\mathcal{A}$ be an algorithm that takes any direction $\vec{w}\in[-1,1]^{d}$ as input and outputs a point $\mathcal{A}(\vec{w})\in P$ such that $ \mathcal{A}(\vec{w})\cdot \vec{w}\geq \alpha \cdot \max_{\vec{x}\in P} \{ \vec{x}\cdot\vec{w}\}$ for some absolute constant $\alpha \le 1$. Then there is a weird separation oracle $WSO$ for $\alpha P$ such that,
\begin{enumerate}
\item Every halfspace output by the $WSO$ will contain  $\alpha P$.
\item Whenever $WSO(\vec{x}) =$ ``yes,'' the execution of $WSO$ explicitly finds directions $\vec{w}_1,\ldots,\vec{w}_l$ such that $\vec{x} \in \text{Conv}\{\mathcal{A}(\vec{w}_1),\ldots, \mathcal{A}(\vec{w}_l)\}$.
\item Let $b$ be the bit complexity of the input vector $\vec{x}$, and $\ell$ be an upper bound of the bit complexity of $\mathcal{A}(\vec{w})$ for all $\vec{w}\in[-1,1]^{d}$, $rt_{\mathcal{A}}(y)$ be the running time of algorithm $\mathcal{A}$ on some input with bit complexity $y$. Then on input $\vec{x}$, $WSO$ terminates in time $\poly\left(d,b,\ell,rt_{\mathcal{A}}(\poly(d,b,\ell))\right)$ and makes at most $\poly(d,b,\ell)$ many queries to $\mathcal{A}$.
\end{enumerate} \end{theorem}

\subsection{Implicit Forms}
Here, we give the necessary preliminaries to understand a mechanism's \emph{implicit form}. The implicit form is oblivious to what allocation rule the mechanism actually uses;  it just stores directly the necessary information to decide if a mechanism is BIC and IR. For a mechanism $M = (A,P)$ and bidder distribution $\mathcal{D}$, the implicit form of $M$ with respect to $\mathcal{D}$ consists of two parts. The first is a function that takes as input a bidder $i$ and a pair of types $t_i,t'_i$, and outputs the expected value of a bidder with type $t_i$ for reporting $t'_i$ instead. Formally, we may store this function as an $m k^2$-dimensional vector $\vec{\pi}(M)$ with:

$$\pi_i(t_i,t'_i) = \mathbb{E}_{\vec{t}_{-i} \leftarrow \mathcal{D}_{-i}}[t_i(A(t'_i;\vec{t}_{-i}))].$$

The second is just a function that takes as input a bidder $i$ and a type $t_i$ and outputs the expected price paid by bidder $i$ when reporting type $t_i$. Formally, we may store this function as a $mk$-dimensional vector $\vec{P}(M)$ with:

$$P_i(t_i) = \mathbb{E}_{\vec{t}_{-i} \leftarrow \mathcal{D}_{-i}}[P_i(t_i;\vec{t}_{-i})].$$

We will denote the implicit form of $M$ as $\vec{\pi}^I(M) = (\vec{\pi}(M), \vec{P}(M))$, and may drop the parameter $M$ where appropriate. We call $\vec{\pi}$ the allocation component of the implicit form and $\vec{P}$ the price component. Sometimes, we will just refer to $\vec{\pi}$ as the implicit form if the context is appropriate.

We say that (the allocation component of) an implicit form, $\vec{\pi}$, is \emph{feasible} with respect to $\mathcal{F},\mathcal{D}$ if there exists a (possibly randomized) mechanism $M$ that chooses an allocation in $\mathcal{F}$ with probability $1$ such that $\vec{\pi}(M) = \vec{\pi}$. We denote by $F(\mathcal{F},\mathcal{D})$ the set of all feasible (allocation components of) implicit forms. We say that an implicit form $\vec{\pi}^I$ is feasible if its allocation component $\vec{\pi}$ is feasible. We say that $\vec{\pi}^I$ is BIC if every mechanism with implicit form $\vec{\pi}^I$ is BIC. It is easy to see that $\vec{\pi}^I$ is BIC if and only if for all $i$, and $t_i,t'_i \in T_i$, we have:

$$\pi_i(t_i,t_i) - P_i(t_i) \geq \pi_i(t_i,t'_i) - P_i(t'_i).$$

Similarly, we say that $\vec{\pi}^I$ is IR if every mechanism with implicit form $\vec{\pi}^I$ is IR. It is also easy to see that $\vec{\pi}^I$ is IR if and only if for all $i$ and $t_i \in T_i$ we have:

$$\pi_i(t_i,t_i) - P_i(t_i) \geq 0.$$

\section{Revenue Maximization}\label{sec:revenue}
In this section, we describe and prove correctness of our reduction when the objective is revenue. {Every result in this section is a special case of our general reduction (that applies to any concave objective) from Section~\ref{sec:general}, and could be obtained as an immediate corollary. We present revenue separately as a special case with the hope that this will help the reader understand the general reduction.} Here is an outline of our approach: In Section~\ref{sec:LPrevenue}, we show that $F(\mathcal{F},\mathcal{D})$ is a convex polytope and write a poly-size linear program that finds the revenue-optimal implicit form provided that we have a separation oracle for $F(\mathcal{F},\mathcal{D})$. In Section~\ref{sec:MDMDPtoSADP} we show that any poly-time $\alpha$-approximation algorithm for {SADP($\mathcal{F}$,$\mathcal{V}$)} implies a poly-time weird separation oracle for {$\alpha F(\mathcal{F},\mathcal{D})$}, and therefore a poly-time $\alpha$-approximation algorithm for MDMDP($\mathcal{F}$, $\mathcal{V}$).

\subsection{Linear Programming Formulation}\label{sec:LPrevenue}
We now show how to write a poly-size linear program to find the implicit form of a mechanism that solves the MDMDP. The idea is that we will search over all feasible, BIC, IR implicit forms for the one that maximizes expected revenue. We first show that $F(\mathcal{F},\mathcal{D})$ is always a convex polytope, then state the linear program and prove that it solves MDMDP. For ease of exposition, most proofs can be found in Appendix~\ref{app:LPrevenue}.

\begin{lemma}\label{lem:convexpolytope} $F(\mathcal{F},\mathcal{D})$ is a convex polytope.
\end{lemma}

\begin{figure}[ht]
\colorbox{MyGray}{
\begin{minipage}{\textwidth} {
\noindent\textbf{Variables:}
\begin{itemize}
\item $\pi_i(t_i,t'_i)$, for all bidders $i$ and types $t_i,t'_i \in T_i$, denoting the expected value obtained by bidder $i$ when their true type is $t_i$ but they report $t'_i$ instead.
\item $P_i(t_i)$, for all bidders $i$ and types $t_i \in T_i$, denoting the expected price paid by bidder $i$ when they report type $t_i$.
\end{itemize}
\textbf{Constraints:}
\begin{itemize}
\item $\pi_i(t_i,t_i) - P_i(t_i) \geq \pi_i(t_i,t'_i) - P_i(t'_i)$, for all bidders $i$, and types $t_i,t'_i \in T_i$, guaranteeing that the implicit form {$(\vec{\pi},\vec{P})$} is BIC.
\item $\pi_i(t_i,t_i) - P_i(t_i) \geq 0$, for all bidders $i$, and types $t_i \in T_i$, guaranteeing that the implicit form {$(\vec{\pi},\vec{P})$} is individually rational.
\item $\vec{\pi} \in F(\mathcal{F},\mathcal{D})$, guaranteeing that the implicit form $(\vec{\pi},\vec{P})$ is feasible.\end{itemize}
\textbf{Maximizing:}
\begin{itemize}
\item $\sum_i \sum_{t_i} \Pr[t_i \leftarrow \mathcal{D}_i] \cdot P_i(t_i)$, the expected revenue when played truthfully by bidders sampled from $\mathcal{D}$.\\
\end{itemize}}
\end{minipage}}
\caption{A linear programming formulation for MDMDP.}
\label{fig:LPMDMD}
\end{figure}
\begin{observation}\label{obs:LPsolvesMDMDP}
Any $\alpha$-approximate solution to the linear program of Figure~\ref{fig:LPMDMD} corresponds to a feasible, BIC, IR implicit form whose revenue is at least a $\alpha$-fraction of the optimal obtainable expected revenue by a feasible, BIC, IR mechanism. 
\end{observation}

\begin{corollary}\label{cor:LP} The program in Figure~\ref{fig:LPMDMD} is a linear program with $\sum_{i\in[m]} (|T_{i}|^2 + |T_{i}|)$ variables. \mattnote{If $b$ is an upper bound on the bit complexity of $Pr[t_i]$ and $t_i(X)$ for all $i,t_i$ and $X \in \mathcal{F}$}, then with black-box access to a weird separation oracle, $WSO$, for  $\alpha F(\mathcal{F},\mathcal{D})$, {the implicit form of an $\alpha$-approximate solution to} MDMDP can be found in time polynomial in $\sum_{i\in[m]} |T_{i}|$, $b$, and the runtime of $WSO$ on inputs with bit complexity polynomial in $\sum_{i\in[m]} |T_{i}|,b$.
\end{corollary}


\subsection{A Reduction from MDMDP to SADP}\label{sec:MDMDPtoSADP}
Based on Corollary~\ref{cor:LP}, the only obstacle to solving the MDMDP is obtaining a separation oracle for $F(\mathcal{F},\mathcal{D})$ (or ``weird'' separation oracle for $\alpha F(\mathcal{F},\mathcal{D})$). In this section, we use Theorem~\ref{thm:CaiDW1} to obtain a weird separation oracle for $\alpha F(\mathcal{F},\mathcal{D})$ using only black box access to an $\alpha$-approximation algorithm for SADP. For ease of exposition, most proofs can be found in Appendix~\ref{app:MDMDPtoSADP}.

In order to apply Theorem~\ref{thm:CaiDW1}, we must first understand what it means to compute $\vec{x} \cdot \vec{w}$ in our setting. Proposition~\ref{prop:virtualwelfare} below accomplishes this. In reading the proposition, recall that $\vec{x}$ is some implicit form $\vec{\pi}$, so the direction $\vec{w}$ has components $w_i(t_i,t'_i)$ for all $i,t_i,t'_i$.  Also note that a type $t_i$ is a function that maps allocations to values. So $\sum_{t_i\in T_{i}} C_i t_i$ is also a function that maps allocations to values (and therefore could be interpreted as a type or virtual type). Namely, it maps $X$ to $\sum_{t_i\in T_{i}} C_i t_i(X)$

\begin{proposition}\label{prop:virtualwelfare}
Let $\vec{\pi} \in F(\mathcal{F},\mathcal{D})$ and let $\vec{w}$ be a direction in $[-1,1]^{\sum_i |T_i|^2}$. Then $\vec{\pi} \cdot \vec{w}$ is exactly the expected virtual welfare of a mechanism with implicit form $\vec{\pi}$ when the virtual type of bidder $i$ with real type $t'_i$ is $\sum_{t_i\in T_{i}} \frac{w_i(t_i,t'_i)}{\Pr[t'_i]}\cdot t_i$.
\end{proposition}

Now that we know how to interpret $\vec{w} \cdot \vec{\pi}$, recall that Theorem~\ref{thm:CaiDW1} requires an algorithm $\mathcal{A}$ that takes as input a direction $\vec{w}$ and outputs a $\vec{\pi}$ with $ \vec{w} \cdot \vec{\pi} \geq \alpha\cdot\max_{\vec{x} \in F(\mathcal{F},\mathcal{D})} \{\vec{w} \cdot \vec{x}\}$. With Proposition~\ref{prop:virtualwelfare}, we know that this is exactly asking for a feasible implicit form whose virtual welfare (computed with respect to $\vec{w}$) is at least an $\alpha$-fraction of the virtual welfare obtained by the optimal feasible implicit form. The optimal feasible implicit form corresponds to a mechanism that, on every profile, chooses the allocation in $\mathcal{F}$ that maximizes virtual welfare. One way to obtain an $\alpha$-approximate implicit form is to use a mechanism that, on every profile, chooses an $\alpha$-approximate outcome in $\mathcal{F}$. Corollary~\ref{cor:obv} below states this formally.

\begin{corollary}\label{cor:obv}
Let $M$ be a mechanism that on profile $(t'_1,\ldots,t'_m)$ chooses a (possibly randomized) allocation $X \in \mathcal{F}$ such that

$$\sum_{i\in[m]} \sum_{t_i\in T_{i}} \frac{w_i(t_i,t'_i)}{\Pr[t'_i]}\cdot t_i(X) \geq \alpha\cdot \max_{X' \in \mathcal{F}} \left\{\sum_{i\in[m]} \sum_{t_i\in T_{i}} \frac{w_i(t_i,t'_i)}{\Pr[t'_i]}\cdot t_i(X')\right\}.$$

Then the implicit form, $\vec{\pi}(M)$ satisfies:

$$\vec{\pi}(M) \cdot \vec{w} \geq \alpha\cdot\max_{\vec{x} \in F(\mathcal{F},\mathcal{D})} \{\vec{x} \cdot \vec{w}\}.$$
\end{corollary}

With Corollary~\ref{cor:obv}, we now want to study the problem of maximizing virtual welfare on a given profile. This turns out to be exactly an instance of SADP. 

\begin{proposition}\label{prop:SADPvirtualwelfare}
Let $t_i \in \mathcal{V}$ for all $i, t_i$. Let also $C_i(t_i)$ be any real numbers, and $\sum_{t_i\in T_{i}} C_i(t_i) t_i(\cdot)$ be the virtual type of bidder $i$. Then any $X \in \mathcal{F}$ that is an $\alpha$-approximation to SADP($\mathcal{F}$,$\mathcal{V}$) on input $(f_1 = \sum_i\sum_{t_i | C_i(t_i) > 0} C_i(t_i) t_i(\cdot), f_2 = \sum_i\sum_{t_i | C_i(t_i) < 0} -C_i(t_i) t_i(\cdot))$ is also an $\alpha$-approximation for maximizing virtual welfare. That is:

$$ \sum_i \sum_{t_i} C_i(t_i) t_i(X) \geq \alpha\cdot\max_{X' \in \mathcal{F}} \left\{\sum_i \sum_{t_i} C_i(t_i) t_i(X')\right\}$$
\end{proposition}

{Combining Corollary~\ref{cor:obv} and Proposition~\ref{prop:SADPvirtualwelfare} yields Corollary~\ref{cor:important} below.}

\begin{corollary}\label{cor:important}
Let $G$ be any $\alpha$-approximation algorithm for SADP($\mathcal{F}$,$\mathcal{V}$). Let also $M$ be the mechanism that, on profile $(t'_1,\ldots,t'_m)$ chooses the allocation $$G\left(\left\{\sum_i \sum_{t_i| w_i(t_i,t'_i) > 0}  (w_i(t_i,t'_i)/\Pr[t'_i])t_i(\cdot)\ ,\ \sum_i \sum_{t_i| w_i(t_i,t'_i) < 0}  -(w_i(t_i,t'_i)/\Pr[t'_i])t_i(\cdot)\right\}\right).$$ Then the interim form $\vec{\pi}(M)$ satisfies:
$$\vec{\pi}(M) \cdot \vec{w} \geq \alpha\cdot\max_{\vec{x} \in F(\mathcal{F},\mathcal{D})} \{\vec{x} \cdot \vec{w}\}.$$
\end{corollary}

At this point, we would like to just let $\mathcal{A}$ be the algorithm that takes as input a direction $\vec{w}$ and computes the implicit form prescribed by Corollary~\ref{cor:important}. Corollary~\ref{cor:important} shows that this algorithm satisfies the hypotheses of Theorem~\ref{thm:CaiDW1}, so we would get a weird separation oracle for $\alpha F(\mathcal{F},\mathcal{D})$. {Unfortunately, this requires some care, as computing the implicit form of a mechanism exactly would require enumerating every profile in the support of $\mathcal{D}$, and also enumerating the randomness used on each profile. Luckily, however, both of these issues arose in previous work and were solved~\cite{CaiDW12b,CaiDW13}. We overview the necessary approach in Section~\ref{app:MDMDPtoSADP}, and refer the reader to~\cite{CaiDW12b,CaiDW13} for complete details.} 

After these modifications, the only remaining step is to turn the implicit form output by the LP of Figure~\ref{fig:LPMDMD} into an actual mechanism. This process is simple and made possible by guarantee 2) of Theorem~\ref{thm:CaiDW1}. We overview the process in Section~\ref{app:MDMDPtoSADP}, as well as give a formal description of our algorithm to solve MDMDP as Algorithm~\ref{alg:solveMDMDP}. We conclude this section with a theorem describing the performance of this algorithm.\mattnote{ In the following theorem, $G$ denotes a (possibly randomized) $\alpha$-approximation algorithm for SADP($\mathcal{F}$,$\mathcal{V}$).}
\begin{theorem}\label{thm:revenue}
Let $b$ be an upper bound on the bit complexity of $t_i(X)$ and $\Pr[t_i]$ for any $i \in [m]$, $t_i \in T_i$, and $X \in \mathcal{F}$. Then Algorithm~\ref{alg:solveMDMDP} makes $\poly(\sum_i |T_i|,1/\epsilon,b)$ calls to $G$, and terminates in time $\poly(\sum_i |T_i|,1/\epsilon, b, rt_{G}(\poly(\sum_i |T_i|,1/\epsilon,b)))$, where $rt_{G}(x)$ is the running time of $G$ on input with bit complexity $x$. If \mattnote{the types are normalized so that $t_i(X) \in [0,1]$ for all $i$, $t_i \in T_i$, and $X \in \mathcal{F}$}, and $OPT$ is the optimal obtainable expected revenue for the given MDMDP instance, then the mechanism output by Algorithm~\ref{alg:solveMDMDP} obtains expected revenue at least $\alpha OPT - \epsilon$, and is $\epsilon$-BIC with probability at least $1-\exp(\poly(\sum_{i}|T_{i}|,1/\epsilon,b))$.\end{theorem}


\section{Reduction from SADP to MDMDP}\label{sec:SADPtoMDMDP}
In this section we overview our reduction from SADP to MDMDP that holds for a certain subclass of SADP instances (a much longer exposition of the complete approach can be found in Section~\ref{app:SADPtoMDMDP}). The subclass is general enough for us to conclude that revenue maximization, even for a single submodular bidder, {is impossible to approximate within any polynomial factor unless $NP = RP$}. For this section, we will restrict ourselves to single-bidder settings, as our reduction will always output a single-bidder instance of MDMDP. 

Here is an outline of our approach: In Section~\ref{sec:compatibility}, we start by defining two properties of {allocation rules. The first of these properties is the well-known \emph{cyclic monotonicity}. The second is a new property we define called \emph{compatibility}. Compatibility is a slightly (and strictly) stronger condition than cyclic monotonicity. The main result of this section is a simple formula (of the form of the objective function that appears in SADP) that upper bounds the maximum obtainable revenue using a given allocation rule, as well as a proof that this bound is attainable when the allocation rule is compatible. Both definitions and results can be found in Section~\ref{sec:compatibility}.
}

{Next, we relate SADP to MDMDP using the results of Section~\ref{sec:compatibility}, showing how to view any (possibly suboptimal) solution to a SADP instance as one for a corresponding MDMDP instance and vice versa. We show that for compatible SADP instances, any optimal solution is also optimal in the corresponding MDMDP instance.  Furthermore, we show (using the work of Section~\ref{sec:compatibility}) that, for any $\alpha$-approximate MDMDP solution $X$, the corresponding SADP solution $Y$  is necessarily an approximate solution to SADP as well, and a lower bound on its approximation ratio as a function of $\alpha$. Therefore, this constitutes a black-box reduction from approximating compatible instances of SADP to approximating MDMDP.} This is presented in Section~\ref{app:SADPtoMDMDP2}.

Finally, in Section~\ref{sec:NPhardSubmodular}, we give a class of compatible SADP instances, where $\mathcal{V}$ is the class of submodular functions and $\mathcal{F}$ is trivial, for which SADP is impossible to approximate within any polynomial factor unless $NP = RP$. Using the reduction of Section~\ref{app:SADPtoMDMDP2} we may immediately conclude that unless $NP = RP$, revenue maximization for a single {monotone} submodular bidder under trivial feasibility constraints (the seller has one copy of each of $n$ goods and can award any subset to the bidder) is impossible to approximate within any polynomial factor. Note that, on the other hand, welfare is trivial to maximize in this setting: simply give the bidder every item. This section is concluded with the proof of the following theorem. Formal definitions of submodularity, value oracle, demand oracle, and explicit access can be found at the start of Section~\ref{sec:NPhardSubmodular}

\begin{theorem}\label{thm:SADPandMDMDPhard}
The problems SADP($2^{[n]}$,monotone submodular functions) (for $k = \poly(n)$) and MDMDP($2^{[n]}$,monotone submodular functions) (for $k = |T_1| = \poly(n)$) are:
\begin{enumerate}
\item Impossible to approximate within any $1/\poly(n)$-factor with only $\poly(k,n)$ value oracle queries.
\item Impossible to approximate within any $1/\poly(n)$-factor with only $\poly(k,n)$ demand oracle queries.
\item Impossible to approximate within any $1/\poly(n)$-factor given explicit access to the input functions in time $\poly(k,n)$, unless $NP = RP$.
\end{enumerate}
\end{theorem}

\section{General Objectives}\label{sec:general}
In this section, we display the full generality of our new approach. In Section~\ref{sec:generalprelim} we update the necessary preliminaries completely. After, we give a very brief overview of how to extend the approach of Section~\ref{sec:revenue} to the general setting. Because the new approach is very similar but technically more involved, we postpone a complete discussion to Appendix~\ref{app:generalappendix}. We include the complete preliminaries to clarify exactly what problem we are solving.

\subsection{Updated Preliminaries}\label{sec:generalprelim}

Here, we update our preliminaries to accommodate the general setting. To ease the transition from the settings described in Section~\ref{sec:prelim}, we list all categories discussed but note only the changes.

\paragraph{Mechanism Design Setting.} No modifications.

\paragraph{Mechanisms.} No modifications.

\paragraph{Goal of the designer.} The designer's goal is to design a feasible, BIC, IR mechanism that maximizes the expected value of some objective function, $\obj$, when encountering a bidder profile sampled from $\mathcal{D}$ (and the bidders report truthfully). $\obj$ takes as input a type profile $\vec{t}$, a randomized outcome $X \in \Delta(\mathcal{F})$, and a randomized payment profile $P$, and it outputs a quality $\obj(\vec{t},X,P)$. We assume that $\obj$ is non-negative and concave with respect to $X$ and $P$. That is, $\obj(\vec{t},cX_1 + (1-c)X_2,cP_1 + (1-c)P_2) \geq c\obj(\vec{t},X_1,P_1) + (1-c)\obj(\vec{t},X_2,P_2)$. As the runtime of our algorithms necessarily depends on the bit complexity of the input, the complexity of $\obj$ must somehow enter the picture as well. One natural notion of complexity is as follows: we know that every concave function can be written as the minimum of a (possibly infinite) set of linear functions. \mattnote{So we shall define the bit complexity of a linear function to be the maximum bit complexity among all its coefficients, and} define the bit complexity of a concave function to be the sum of the bit complexities of each of these linear functions.\footnote{This is the natural complexity of an explicit description of a concave function that lists each such linear function. All objectives considered in this paper have finite bit complexity under this definition, but note that this value could well be infinite for arbitrary concave functions. Such objectives can be accommodated using techniques from convex optimization, but doing so carefully would cause the exposition to become overly cumbersome.}

To help put this in context, here are some examples: the social welfare objective can be described as $\obj(\vec{t},X,P) = \sum_i t_i(X)$. The revenue objective can be described as $\obj(\vec{t},X,P) = \sum_i \mathbb{E}[P_i]$. The fractional max-min objective can be described as $\obj(\vec{t},X,P) = \min_i \{t_i(X)\}$. All three examples are concave. We say that $\obj$ is \emph{allocation-only} if $\obj$ does not depend on $P$, and that $\obj$ is \emph{price-only} if $\obj$ does not depend on $\vec{t}$ or $X$. For instance, social welfare and fractional max-min are both allocation-only objectives, and revenue is a price-only objective. With the exception of revenue, virtually every objective function studied in the literature (in both mechanism and algorithm design) is allocation-only.\footnote{ A recent example of an objective function that considers both the allocation and pricing scheme can be found in~\cite{BGM13}.} For allocation-only or price-only objectives, the necessary problem statements and results have a cleaner form, which we state explicitly.

\paragraph{Formal Problem Statements.} In the problem statements below, MDMDP and SADP have been updated to contain an extra parameter $\obj$, which denotes the objective function. For SADP, we also specify simplifications for allocation-only and price-only objectives.

\notshow{\textbf{MDMDP($\mathcal{F}$, $\mathcal{V}$, $\obj$):} Given as input types $t_{ij} \in \mathcal{V}$ ($1 \leq i \leq m$, $1 \leq j \leq |T_{i}|$), distributions $\mathcal{D}_i$ over $T_i = \{t_{ij}| j\}$, find a feasible (outputs an allocation in $\mathcal{F}$ with probability $1$) BIC, IR mechanism, $M$, that maximizes $\obj$ in expectation, when $m$ bidders with types sampled from $\mathcal{D} = \times_i \mathcal{D}_i$ play $M$ truthfully (with respect to all feasible, BIC, IR mechanisms). $M$ is said to be an $\alpha$-approximation to MDMDP if the expected value of $\obj$ is at least an $\alpha$-fraction of the optimal obtainable expected value of $\obj$.}

\smallskip \textbf{MDMDP($\mathcal{F}$,$\mathcal{V}$,$\obj$):} {\sc Input:} For each bidder $i \in [m]$, a finite set of types $T_i \subseteq {\cal V}$ and a distribution ${\cal D}_i$ over $T_i$. {\sc Goal:} Find a feasible (outputs an outcome in $\mathcal{F}$ with probability $1$) BIC, IR mechanism~$M$, that maximizes $\obj$ in expectation, when $m$ bidders with types sampled from $\mathcal{D} = \times_i \mathcal{D}_i$ play $M$ truthfully (with respect to all feasible, BIC, IR mechanisms). $M$ is said to be an $\alpha$-approximation to MDMDP if the expected value of $\obj$ is at least an $\alpha$-fraction of the optimal obtainable expected value of $\obj$.

\textbf{SADP($\mathcal{F}$, $\mathcal{V}$, $\mathcal{O}$):} Given as input functions$f_{j} \in \mathcal{V}^{*}$ ($1 \leq j \leq k$), $g_i \in \mathcal{V}$ ($1\leq i \leq m$),  multipliers $c_i \in \mathbb{R}$ ($1 \leq i \leq m$), and a special multiplier $c_0 \geq 0$, find a feasible (possibly randomized) outcome $X \in \Delta(\mathcal{F})$ and (possibly randomized) pricing scheme $P$ such that there exists an index $j^* \in [k-1]$ with:
\begin{align*}
&\left(c_0 \cdot \obj((g_1,\ldots,g_m),X,P)\right) + \left(\sum_i c_i \mathbb{E}[P_i]\right) + \left(f_{j^*}(X) - f_{j^*+1}(X)\right) 
\\= &\max_{X' \in \mathcal{F},P} \left\{\left(c_0 \cdot \obj((g_1,\ldots,g_m),X',P)\right) + \left(\sum_i c_i \mathbb{E}[P_i]\right) + \left(f_{j^*}(X') - f_{j^*+1}(X')\right)\right\}.
\end{align*}
$X$ is said to be an $\alpha$-approximation to SADP if there exists an index $j^* \in [k-1]$ such that:

\begin{align*}
&\left(c_0 \cdot \obj((g_1,\ldots,g_m),X,P)\right) + \left(\sum_i c_i \mathbb{E}[P_i]\right) + \left(f_{j^*}(X) - f_{j^*+1}(X)\right) 
\\\geq &\alpha \left( \max_{X' \in \mathcal{F},P} \left\{\left(c_0 \cdot \obj((g_1,\ldots,g_m),X',P)\right) + \left(\sum_i c_i \mathbb{E}[P_i]\right) + \left(f_{j^*}(X') - f_{j^*+1}(X')\right)\right\}\right).
\end{align*}

One should interpret this objective function as maximizing $\obj$ plus virtual revenue plus virtual welfare. For allocation-only objectives, we may remove the price entirely (from the input as well as objective). That is, the input is just $f_j \in \mathcal{V}^*$ ($1\leq j \leq k$) and $g_i \in \mathcal{V}$ ($1 \leq i \leq m$) and the objective is $c_0 \obj((g_1,\ldots,g_m),X) + (f_{j^*}(X) - f_{j^*+1}(X))$. This should be interpreted as maximizing $\obj$ plus virtual welfare.

For price-only objectives, we may remove the $g_i$'s entirely (from the input as well as objective). That is, the input is just $f_j \in \mathcal{V}^*$ ($1 \leq j \leq k$) and multipliers $c_i \in \mathbb{R}$ ($1 \leq i \leq m$) and the objective is $c_0 \obj(P) + \left(\sum_i c_i\mathbb{E}[P_i]\right) + \left(f_{j^*}(X) - f_{j^*+1}(X)\right)$. It is easy to see that this separates into two completely independent problems, as the allocation $X$ and pricing scheme $P$ don't interact. The first problem is just finding a pricing scheme that maximizes $\obj(P) + \sum_i c_i\mathbb{E}[P_i]$, which should be interpreted as maximizing $\obj$ plus virtual revenue. The second is solving the simple version of SADP given in Section~\ref{sec:prelim}, which should be interpreted as maximizing virtual welfare. When the objective is revenue, maximizing $\obj(P) + \sum_i c_i\mathbb{E}[P_i]$ is trivial, {simply charge the maximum allowed price to each bidder with $c_i > -1$ and $0$ to everyone else}. This is why solving the general version of SADP for revenue reduces to the simpler version given in Section~\ref{sec:prelim}. 

\subsection{Implicit Forms}
We update the implicit form only by adding one extra component, which stores the expected value of $\obj$ when bidders sampled from $\mathcal{D}$ play $M$ truthfully. That is, we define:

$$\pi_\obj(M) = \mathbb{E}_{\vec{t} \leftarrow \mathcal{D}}[\obj(\vec{t},A(\vec{t}),P(\vec{t}))].$$

We still denote the implicit form of $M$ as $\vec{\pi}^I(M) = (\pi_\obj(M), \vec{\pi}(M), \vec{P}(M))$. We call $\pi_\obj(M)$ the objective component of $\vec{\pi}^I$. Again, we will sometimes just refer to $(\pi_\obj(M), \vec{\pi}(M))$ as the implicit form if the context is appropriate.

We modify slightly our definition of feasibility for implicit forms (but this is important). We say that an implicit form $\vec{\pi}^I=(\pi_{\obj},\vec{\pi},\vec{P})$ is \emph{feasible} with respect to $\mathcal{F},\mathcal{D},\obj$ if there exists a (possibly randomized) mechanism $M$ that chooses an allocation in $\mathcal{F}$ with probability $1$ such that $\vec{\pi}(M) = \vec{\pi}$, $\vec{P}(M) = \vec{P}$, and $\pi_\obj(M) \geq \pi_\obj \geq 0$.\footnote{For the feasible region to be convex, it is important to include all $\obj$ that is smaller than $\pi_\obj(M)$. Also, as eventually we need to maximize $\obj$, including these points will not affect the final solution. We bound it below by $0$ in order to maintain a bounded feasible region.} We denote by $F(\mathcal{F},\mathcal{D},\obj)$\footnote{\mattnote{Assuming that all types have been normalized so that $t_i(X) \in [0,1]$ for all $i,t_i \in T_i$, $X \in \mathcal{F}$,} we also may w.l.o.g. add the vacuous constraint that $0 \leq P_i(t_i) \leq 1$ for all $i,t_i$ in order to keep the feasible region bounded. These constraints are vacuous because they will be enforced in the end anyway by individual rationality.} the set of all feasible implicit forms. For allocation-only objectives, we may ignore the price component, and discuss feasibility only of $(\pi_\obj,\vec{\pi})$ (still calling an implicit form $\vec{\pi}^I$ feasible if its objective and allocation component are feasible). \mattnote{For allocation-only objectives, we therefore denote by $F(\mathcal{F},\mathcal{D},\mathcal{O})$ the set of all feasible (objective and allocation components of) implicit forms}. For price-only objectives, we can separate the question of feasibility into two parts: one for $\vec{\pi}$, and the other for $(\pi_\obj,\vec{P})$. The implicit form $\vec{\pi}^I$ is feasible if and only if $\vec{\pi}$ is feasible and $(\pi_\obj,\vec{P})$ is feasible (this is not a definition, but an observation). \mattnote{Therefore, for price-only objectives, we denote by $F(\mathcal{F},\mathcal{D},\obj)$ the set of all feasible (objective and price components of) implicit forms, and by $F(\mathcal{F},\mathcal{D})$ the set of all feasible allocation components of implicit forms.}
\subsection{Overview of Approach}\label{sec:generaloverview}
Here we give a very brief overview of how to modify the approach of Section~\ref{sec:revenue} to accommodate the general setting just described in Section~\ref{sec:generalprelim}. In Section~\ref{sec:LPgeneral}, we modify the linear program of Section~\ref{sec:MDMDPtoSADP} to search over all feasible, BIC, IR implicit forms in the general setting for the one that maximizes $\obj$ in expectation. This linear program is still of polynomial size, but requires a separation oracle for $F(\mathcal{F},\mathcal{D},\obj)$ (or a weird separation oracle for $\alpha F(\mathcal{F},\mathcal{D},\obj)$). 

In Section~\ref{sec:MDMDPtoSADPgeneral}, we use Theorem~\ref{thm:CaiDW1} to obtain a weird separation oracle for $\alpha F(\mathcal{F},\mathcal{D},\obj)$. We show that $\vec{x} \cdot \vec{w}$ corresponds to the expected value of a virtual objective for a mechanism $M$ with interm form $\vec{x}$. For general objectives that consider both price and allocation rules, this virtual objective can be interpreted as maximizing the actual objective plus virtual welfare plus virtual revenue. For allocation-only objectives, it can be interpreted as maximizing the actual objective plus virtual welfare. For price-only objectives, it can be interpreted as separately maximizing the actual objective plus virtual revenue, and separately maximizing virtual welfare. The section concludes with a proof of the following theorem. The referenced Algorithm~\ref{alg:solveMDMDPgeneral} can be found in Section~\ref{sec:MDMDPtoSADPgeneral}.\mattnote{ In the following theorem, $G$ denotes a (possibly randomized) $\alpha$-approximation algorithm for SADP($\mathcal{F}$,$\mathcal{V}$,$\obj$).}

\begin{theorem}\label{thm:objective}
Let $b$ be an upper bound on the bit complexity of \mattnote{$\obj$}, $\Pr[t_i]$, and $t_i(X)$ for all $i,t_i \in T_i, X \mattnote{\in {\cal F}}$. Then Algorithm~\ref{alg:solveMDMDPgeneral} makes $\poly(\sum_i |T_i|,1/\epsilon,b)$ calls to $G$, and terminates in time \\$\poly(\sum_i |T_i|,1/\epsilon, b, rt_{G}(\poly(\sum_i |T_i|,1/\epsilon, b)))$, where $rt_{G}(x)$ is the running time of $G$ on an input with bit complexity $x$. If the range of each $t_{i}$ is normalized to lie in $[0,1]$ for all $i, t_i \in T_i$, and $OPT$ is the optimal obtainable expected value of $\obj$ for the given MDMDP instance, then the mechanism output by Algorithm~\ref{alg:solveMDMDP} yields $\mathbb{E}[\obj] \geq \alpha OPT - \epsilon$, and is $\epsilon$-BIC with probability at least $1-\exp(\poly(\sum_{i}|T_{i}|,1/\epsilon,b))$.
\end{theorem}

\section{Fractional Max-Min Fairness}\label{sec:maxmin}
In this section, we apply our results to a concrete objective function: fractional max-min fairness {$\obj(\vec{t},{X}) = \min_i \{t_i({X})\}$}. For the remainder of this section we will denote this objective by FMMF. We will also restrict ourselves to feasibility constraints and bidder types considered in~\cite{CaiDW12b}. Specifically, there are going to be $n$ items and $\mathcal{F}$ will be some subset of $2^{[m] \times [n]}$, where the element $(i,j)$ denotes that bidder $i$ receives item $j$. Bidders will be additive, meaning that we can think of $t_i$ as an $n$-dimensional vector $\vec{t}_i$ such that, if $x_{ij}$ denotes the probability that bidder $i$ receives item $j$ in {allocation ${X}$}, then {$t_i({X}) = \sum_j t_{ij}\cdot x_{ij}$}. {For the rest of this section, we will use $\vec{x}$ to represent the vector of marginal probabilities  $x_{ij}$ induced by allocation $X$, and $F(\mathcal{F})$ to represent the set of marginal probability vectors induced by some feasible allocation.} Morever, we can take ${\cal V}=[0,1]^{n \times m}$ to represent the set of all additive functions over ${\cal F}$ and think of bidder $i$'s type as an element of ${\cal V}$ which is $0$ on all $(i',j')$ except for those satisfying $i'=i$. For the rest of this section, if $f \in {\cal V}$, we write $f_{ij}$ as a shorthand for the $(i,j)$ coordinate of $f$. Our main result will be a black-box reduction from SADP$({\cal F},{\cal V},FMMF)$ to a very well-studied problem that we call Max-Weight($\mathcal{F}$)\footnote{Notice that FMMF is in fact well defined for all vectors $\vec{f}$ of functions $f^{(i)} \in {\cal V}$. We will define FMMF of an allocation $X$ as $\min_{i}\{ f^{(i)}(X)=\sum_{\ell,j}f^{(i)}_{\ell,j}\cdot x_{\ell j}$\}.}. Using this reduction we will obtain optimal and approximately optimal mechanisms for a wide range of ${\cal F}$'s.

\medskip \textbf{Max-Weight($\mathcal{F}$):} {\sc Input:} $w_{ij} \in \mathbb{R}$ for all $i \in [m],j \in [n]$. {\sc Goal:} Find $S \in \mathcal{F}$ that maximizes $\sum_{(i,j) \in S}w_{ij}$. $S$ is said to be an $\alpha$-approximation to Max-Weight if $\sum_{(i,j) \in S}w_{ij} \geq \alpha \max_{S' \in \mathcal{F}} \{\sum_{(i,j) \in S'}w_{ij}\}$.

\smallskip Some popular examples of Max-Weight($\mathcal{F}$) include finding the max-weight matching, the max-weight independent set in a matroid, the max-weight feasible set in a matroid intersection, etc. Our high-level goal is to obtain a poly-time reduction from SADP($\mathcal{F}$, additive functions, $FMMF$) to  Max-Weight($\mathcal{F})$, and then obtain mechanisms for FMMF. 

We first need to show that $FMMF$ is concave. The proof of the next observation, as well as all other proofs of this section, are given in Section~\ref{app:maxmin}.

\begin{observation}\label{obs:concave}
$FMMF$ is concave.
\end{observation}

We now begin the description of our algorithm for SADP. We first observe that any additive function $f$ over ${\cal F}$ can be described as an $nm$-dimensional vector, so that $f(X) = \sum_{i,j} f_{ij}\cdot x_{ij}$. It is then easy to see that $f(X) - f'(X) = \sum_{i,j} (f_{ij} - f'_{ij}) \cdot x_{ij}$. So let's define $\obj'(\vec{g},X,f,f') = c_0 \cdot FMMF(\vec{g},X) + f(X) - f'(X)$, where $\vec{g}=(g^{(1)},\ldots,g^{(m)}) \in {\cal V}^m$, and $f, f' \in {\cal V}$ are arbritrary additive functions over ${\cal F}$. Notice that, if we can optimize ${\cal O}'$ then we can also optimize SADP. (Indeed, we can optimize SADP in a strong sense in that we can solve for all $j^*$, not just one.) To do so, we make use of the linear programming formulation of Figure~\ref{fig:LPmaxmin}, which aims at maximizing $\obj'(\vec{g},X,f,f')$ over all $\vec{x} \in F(\mathcal{F})$. (It is clear that $\Delta(\mathcal{F})$ is a convex polytope, as a convex hull of elements in $\mathcal{F}$. As $F(\mathcal{F})$ is simply a linear transformation of $\Delta(\mathcal{F})$, it is also a convex polytope.)

\begin{figure}[ht]
\colorbox{MyGray}{
\begin{minipage}{\textwidth} {
\noindent\textbf{Variables:}
\begin{itemize}
\item $x_{ij}$, for $1 \leq i \leq m$, $1 \leq j \leq n$, denoting the probability that $X$ awards item $j$ to bidder $i$.
\item $O$, denoting a lower bound on $FMMF(\vec{g},X)$.
\end{itemize}
\textbf{Constraints:}
\begin{itemize}
\item $\vec{x}\in F(\mathcal{F})$, guaranteeing that $\vec{x}$ corresponds to a valid solution.
\item $O \leq \sum_{\ell, j} g^{(i)}_{\ell j} x_{\ell j}$, for all $i \in [m]$, guaranteeing that $O \leq FMMF(\vec{g},X)$.
\end{itemize}
\textbf{Maximizing:}
\begin{itemize}
\item $c_0O + \sum_{i,j} (f_{ij} - f'_{ij})x_{ij}$, a lower bound on $\obj'(\vec{g},X,f,f')$.\\
\end{itemize}}
\end{minipage}}
\caption{A linear programming formulation for maximizing $\obj'(\vec{g},X,f,f')$.}
\label{fig:LPmaxmin}
\end{figure}

\begin{observation}\label{obs:LPsolvesmaxmin}
An $\alpha$-approximate solution to the LP of figure~\ref{fig:LPmaxmin} corresponds to an $\alpha$-approximate solution to maximizing $\obj'(\vec{g},X,f,f')$ over $\Delta(\mathcal{F})$.
\end{observation}

\begin{corollary}\label{cor:LPmaxmin} The program in Figure~\ref{fig:LPmaxmin} is a linear program with $mn+1$ variables. If $b$ is an upper bound on the bit complexity of $g_{ij},f_{ij}$, and $f'_{ij}$ for all $i,j$, then with black-box access to a weird separation oracle, $WSO$, for  {$\alpha F(\mathcal{F})$,} an $\alpha$-approximate solution to maximizing $\obj'(\vec{g},X,f,f')$ over {$F(\mathcal{F})$} can be found in time polynomial in $n$, $m$, $b$ and the runtime of $WSO$ on inputs with bit complexity polynomial in $n$, $m$, and $b$.
\end{corollary}

With Corollary~\ref{cor:LPmaxmin}, our task is  to now come up with a weird separation oracle for {$\alpha F(\mathcal{F})$}, which we show can be obtained using black-box access to an $\alpha$-approximation algorithm for Max-Weight($\mathcal{F}$). Again, we would like to apply Theorem~\ref{thm:CaiDW1}, so we just need an algorithm that can take as input a direction $\vec{w} \in [-1,1]^{nm}$ and output an {$\vec{x}$ with $\vec{w} \cdot \vec{x} \geq \alpha \max_{\vec{x}' \in F(\mathcal{F})}\{\vec{w} \cdot \vec{x}'\}$}. Fortunately, this is much easier than the previous sections. \mattnote{Below, we recall that each direction $\vec{x}$ has a component in $[0,1]$ for all $i,j$, which we interpret as the marginal probability that the element $(i,j)$ is chosen by a corresponding distribution $X$}.

\begin{observation}\label{obs:obvious}
{$\vec{w} \cdot \vec{x}$} is exactly the expected weight of a set sampled from $X$, where the weight of a set $S$ is $\sum_{(i,j) \in S} w_{ij}$.
\end{observation}

\begin{corollary}\label{cor:obvious}
Let $G'$ be any deterministic $\alpha$-approximation algorithm for Max-Weight($\mathcal{F}$) and  define $G'(\vec{w})_{ij}$ to be $1$ if $(i,j) \in G'(\vec{w})$ and $0$ otherwise. Then $G'(\vec{w}) \cdot {\vec{w}} \geq {\alpha} \max_{\vec{x} \in F(\mathcal{F})}\{\vec{x} \cdot \vec{w}\}$.
\end{corollary}

If we want to let our reduction accommodate randomized algorithms for Max-Weight($\mathcal{F})$, we need to use the reduction of~\cite{CaiDW13} from randomized to deterministic mechanisms. Recall that the modification is basically this: if $G$ is a randomized $\alpha$-approximation algorithm, then running $G$ independently enough times and taking the best solution will give an $(\alpha - \gamma)$-approximation with very high probability for any desired $\gamma>0$. If we call this algorithm $G'$, then we may treat $G'$ as a deterministic algorithm for all intents and purposes by fixing the randomness used by $G'$ ahead of time, and taking a union bound to guarantee that with very high probability, $G'$ obtains an $(\alpha - \gamma)$-approximation on all instances that ever arise during the problem.

\begin{algorithm}[ht]
        \caption{Algorithm to solve SADP($\mathcal{F}$,additive functions,$FMMF$):}
    \begin{algorithmic}[1]\label{alg:solveSADPmaxmin}
        \STATE Input: additive functions $\{g_{i}\}_{i\in[m]}$, $\{f_j\}_{j \in [k]}$, and black-box access to $G$, a (possibly randomized) $\alpha$-approximation algorithm for Max-Weight($\mathcal{F}$).
        \STATE Ignore $f_3,\ldots,f_k$. Only $f_1$ and $f_2$ will be used.
        \STATE Sample $G$ to obtain $G'$ as prescribed in Appendix~G.2 of~\cite{CaiDW13} (and intuitively explained above). $G'$ can be henceforth treated as a deterministic $(\alpha-\gamma)$-approximation  algorithm, for any desired~$\gamma$.
        \STATE Define $\mathcal{A}$ to, on input $\vec{w} \in [-1,1]^{nm}$, output the set $G'(\vec{w})$ as a vector in $\{0,1\}^{nm}$.
         \STATE Solve the linear program in Figure~\ref{fig:LPmaxmin} to maximize $\obj'(\vec{g},X,f_1,f_2)$ using the weird separation oracle obtained from $\mathcal{A}$ using Theorem~\ref{thm:CaiDW1}. Call the output {$\vec{x}$}, and let $\vec{w}_1,\ldots,\vec{w}_l$ be the directions guaranteed by Theorem~\ref{thm:CaiDW1} to be output by the weird separation oracle.
         \STATE Solve the linear system {$\vec{x}= \sum_{j\in[l] }c_j \mathcal{A}(\vec{w}_j)$} to write {$\vec{x}$} as a convex combination of $\mathcal{A}(\vec{w}_j)$.
         \STATE Output the following distribution: First, randomly select a $\vec{w} \in [-1,1]^{nm}$ by choosing $\vec{w}_j$ with probability $c_j$. Then, choose the set $G'(\vec{w})$.
             \end{algorithmic}
\end{algorithm}

\begin{theorem}\label{thm:maxmin}
Let $b$ be an upper bound on the bit complexity {of the values in the range of the functions input to Algorithm~\ref{alg:solveSADPmaxmin}}. Then Algorithm~\ref{alg:solveSADPmaxmin} makes $\poly(n,m,b)$ calls to $G$ and terminates in time $\poly(n,m,b)$. The distribution output is an $\alpha$-approximation to the input SADP instance.
\end{theorem}

\begin{corollary}\label{cor:maxmin}
Given black-box access to $G$, an $\alpha$-approximation algorithm for Max-Weight($\mathcal{F}$), an $\epsilon$-BIC solution MDMDP($\mathcal{F}$,additive functions,$FMMF$) can be found with expected fractional max-min fairness at least $\alpha OPT - \epsilon$. If $b$ is an upper bound on the bit complexity of $t_{ij}$ and $\Pr[t_i]$ for all $i,j, t_i$, then the algorithm runs in time $\poly(b,n,\sum_i |T_i|)$ and makes $\poly(b,n,\sum_i |T_i|)$ calls to $G$ on inputs of size $\poly(b,n,\sum_i |T_i|)$.
\end{corollary}

{\begin{remark}
Notice that if ${\cal F}$ is a matroid or an intersection of two matroids, then we can solve Max-Weight($\mathcal{F}$) exactly. As a consequence of Corollary~\ref{cor:maxmin} we obtain optimal mechanisms for the Fractional Max-Min fairness objective for such ${\cal F}$'s.
\end{remark}}

\appendix
\section{Omitted Formal Details from Section~\ref{sec:prelim}}\label{app:prelim}
\begin{definition}(BIC) A mechanism $M = (A,P)$ is said to be BIC if for all bidders $i$ and all types $t_i, t'_i \in T_i$ we have:
$$\mathbb{E}_{t_{-i} \leftarrow \mathcal{D}_{-i}} \left[t_i(A(t_i;t_{-i})) - P_i(t_i;t_{-i})\right] \geq \mathbb{E}_{t_{-i} \leftarrow \mathcal{D}_{-i}} \left[t_i(A(t'_i;t_{-i})) - P_i(t'_i;t_{-i})\right].$$

\end{definition}

\begin{definition} (IR) A mechanism $M = (A,P)$ is said to be (interim) IR\footnote{We note briefly here that for all objectives considered in this paper, any BIC, interim IR mechanism can be made to be BIC and ex-post IR without any loss via a simple reduction given in~\cite{DW12}.} if for all bidders $i$ and all types $t_i \in T_i$ we have:
$$\mathbb{E}_{t_{-i} \leftarrow \mathcal{D}_{-i}} \left[t_i(A(t_i;t_{-i})) - P_i(t_i;t_{-i})\right] \geq 0.$$
\end{definition}

\subsection{Linear Programming}\label{app:linearprogramming}
Here, we review the necessary preliminaries on linear programming. {Theorem~\ref{thm:CaiDW1} comes from recent work~\cite{CaiDW13}. Theorem~\ref{thm:ellipsoid} states well-known properties of the ellipsoid algorithm. \mattnote{Corollary~\ref{cor:CaiDW1} is an obvious corollary of part 1 of Theorem~\ref{thm:CaiDW1}. In addition, a complete discussion of this can be found in~\cite{CaiDW13}.}

\begin{corollary}\label{cor:CaiDW1}(\cite{CaiDW13}) Let $Q$ be an arbitrary intersection of halfspaces. Let $SO$ be a separation oracle for $\alpha P$, where $P$ is a bounded convex polytope containing the origin and $\alpha \le 1$ some constant. Let $c_1$ be the \mattnote{solution output by the Ellipsoid algorithm} that maximizes some linear objective $\vec{c} \cdot \vec{x}$ subject to $\vec{x} \in Q$ and $SO(\vec{x}) = ``yes''$. Let also $c_2$ be the solution \mattnote{output by the exact same algorithm}, but replacing $SO$ with $WSO$, a ``weird'' separation oracle for $\alpha P$ as in Theorem~\ref{thm:CaiDW1}---i.e. \mattnote{run the Ellipsoid algorithm with the exact same parameters} as if $WSO$ was a valid separation oracle for $\alpha P$. Then $c_2 \geq c_1$.
\end{corollary}

\begin{theorem}\label{thm:ellipsoid}[Ellipsoid Algorithm for Linear Programming]
Let $P$ be a bounded convex polytope in $\mathbb{R}^d$ specified via a separation oracle $SO$, and let $\vec{c}\cdot\vec{x}$ be a linear function. Suppose that $\ell$ \mattnote{is an upper bound on the bit complexity of the coordinates of $\vec{c}$ as well as the extreme points of $P$},\footnote{\mattnote{If a $d$-dimensional convex region with extreme points of bit complexity $\ell$ is non-empty, then it certainly has volume at least $2^{-\poly(d,\ell)}$.}} and also that we are given a ball $B(x_0, R)$ containing $P$ such that $x_0$ and $R$ have bit complexity $\poly(d,\ell)$. Then we can run the ellipsoid algorithm to optimize $\vec{c}\cdot\vec{x}$ over $P$, maintaining the following properties:
\begin{enumerate}
\item The algorithm will only query $SO$ on rational points with bit complexity $\poly(d,\ell)$.
\item The ellipsoid algorithm will solve the Linear Program in time polynomial in $d$, $\ell$ and the runtime of $SO$ when the input query is a rational point of bit complexity $\poly(d,\ell)$.
\item The output optimal solution is a corner of $P$.
\end{enumerate}
\end{theorem}


\section{Reduction from SADP to MDMDP (complete)}\label{app:SADPtoMDMDP}

\subsection{Cyclic Monotonicity and Compatibility}\label{sec:compatibility}

We start with two definitions that discuss matching types to allocations. By this, we mean to form a bipartite graph with types on the left and (possibly randomized) allocations on the right. The weight of an edge between a type $t$ and (possibly randomized) allocation $X$ is exactly $t(X)$. So the weight of a matching in this graph is just the total welfare obtained by awarding each allocation to its matched type. So when we discuss the welfare-maximizing matching of types to allocations, we mean the max-weight matching in this bipartite graph. Below, cyclic monotonicity is a well-known definition in mechanism design with properties connected to truthfulness. Compatibility is a new property that is slightly stronger than cyclic monotonicity.

\begin{definition}(Cyclic Monotonicity) A list of (possibly randomized) allocations $(X_1,\ldots,X_k)$ is said to be cyclic monotone with respect to $(t_1,\ldots,t_k)$ if the welfare-maximizing matching of types to allocations is to match allocation $X_i$ to type $t_i$ for all $i$. 
\end{definition}

\begin{definition}(Compatibility) We say that a list of types $(t_1,\ldots,t_k)$ and a list of (possibly randomized) allocations $(X_1,\ldots,X_k)$ are \emph{compatible} if $(X_1,\ldots,X_k)$ is cyclic monotone with respect to $(t_1,\ldots,t_k)$, and for any $i < j$, the welfare-maximizing matching of types $t_{i+1},\ldots,t_{j}$ to (possibly randomized) allocations $X_i,\ldots,X_{j-1}$ is to match allocation $X_\ell$ to type $t_{\ell+1}$ for all $\ell$.
\end{definition}

Compatibility is a slightly stronger condition than cyclic monotonicity. In addition to the stated definition, it is easy to see that cyclic monotonicity also guarantees that for all $i < j$, the welfare-maximizing matching of types $t_{i+1},\ldots,t_j$ to allocations $X_{i+1},\ldots,X_j$ is to match allocation $X_\ell$ to type $t_{\ell}$ for all $\ell$. Compatibility requires that the ``same'' property still holds if we shift the allocations down by one. 

Now, we want to understand how the type space and allocations relate to expected revenue. Below, $\vec{t}$ denotes an ordered list of $k$ types, $\vec{q}$ denotes an ordered list of $k$ probabilities and $\vec{X}$ denotes an ordered list of $k$ (possibly randomized) allocations. We denote by $Rev(\vec{t},\vec{q},\vec{X})$ the maximum obtainable expected revenue in a BIC, IR mechanism that awards allocation $X_j$ to type $t_j$, and the bidder's type is $t_j$ with probability $q_j$. Proposition~\ref{prop:upperboundrevenue} provides an upper bound on $Rev$ that holds in every instance. Proposition~\ref{prop:exactrevenue} provides sufficient conditions for this bound to be tight.

\begin{proposition}\label{prop:upperboundrevenue} For all {$\vec{t},\vec{q},\vec{X}$}, we have:
\begin{equation}\label{eq:revenue}
Rev(\vec{t},\vec{q},\vec{X}) \leq \sum_{\ell=1}^{k}\sum_{j \geq \ell} q_jt_\ell(X_\ell) - \sum_{\ell = 1}^{k-1}\sum_{j \geq \ell+1}q_jt_{\ell+1}(X_\ell)
\end{equation}
\end{proposition}

\begin{prevproof}{Proposition}{prop:upperboundrevenue}
Let $p_j$ denote the price charged in some BIC, IR mechanism that awards allocation $X_j$ to type $t_j$. Then IR guarantees that:

$$p_1 \leq t_1(X_1)$$

Furthermore, BIC guarantees that, for all $j \leq k$:

\begin{align*}
t_j(X_j) - p_j \geq t_j(X_{j-1}) - p_{j-1}\\
\Rightarrow p_j \leq t_j(X_j) -t_j(X_{j-1}) + p_{j-1}
\end{align*}

Chaining these inequalities together, we see that, for all $j \leq k$:

$$p_j \leq t_1(X_1) + \sum_{\ell = 2}^j \left(t_\ell(X_\ell) - t_{\ell}(X_{\ell-1})\right)$$

As the expected revenue is exactly $\sum_j p_j q_j$, we can rearrange the above inequalities to yield:

\begin{align*}
\sum_j p_jq_j &\leq t_1(X_1) \sum_j q_j +\sum_{\ell=2}^k \sum_{j \geq \ell} q_j\left(t_\ell(X_\ell) - t_\ell(X_{\ell-1}) \right)\\
&\leq t_1(X_1) + \sum_{\ell = 2}^k \sum_{j \geq \ell} q_j t_\ell(X_\ell) - \sum_{\ell = 2}^k \sum_{j \geq \ell} q_j t_\ell(X_{\ell-1})\\
&\leq \sum_{\ell = 1}^k \sum_{j \geq \ell} q_j t_\ell(X_\ell) - \sum_{\ell = 1}^{k-1} \sum_{j \geq \ell+1} q_j t_{\ell+1}(X_{\ell})
\end{align*}

\end{prevproof}

\begin{proposition}\label{prop:exactrevenue} If $\vec{t}$ and $\vec{X}$ are compatible, then for all $\vec{q}$, Equation~\eqref{eq:revenue} is tight. That is, 
$$Rev(\vec{t},\vec{w},\vec{X}) = \sum_{\ell=1}^{k}\sum_{j \geq \ell} q_jt_\ell(X_\ell) - \sum_{\ell = 1}^{k-1}\sum_{j \geq \ell+1}q_jt_{\ell+1}(X_\ell)$$
\end{proposition}

\begin{prevproof}{Proposition}{prop:exactrevenue}
Proposition~\ref{prop:upperboundrevenue} guarantees that the maximum obtainable expected payment does not exceed the right-hand bound, so we just need to give payments that yield a BIC, IR mechanism whose expected payment is as desired. So consider the same payments used in the proof of Proposition~\ref{prop:upperboundrevenue}:

\begin{align*}
p_1 &= t_1(X_1)\\
p_j &= t_j(X_j) - t_j(X_{j-1}) + p_{j-1},\ j \geq 2\\
&= t_1(X_1) + \sum_{\ell = 2}^j (t_\ell(X_\ell) - t_{\ell}(X_{\ell-1}))
\end{align*}

The exact same calculations as in the proof of Proposition~\ref{prop:upperboundrevenue} shows that the expected revenue of a mechanism with these prices is exactly the right-hand bound. So we just need to show that these prices yield a BIC, IR mechanism. For simplicity in notation, let $p_0 = 0$, $t_0 = 0$ (the $0$ function), and $X_0$ be the null allocation (for which we have $t_j(X_0) = 0\ \forall j$). To show that these prices yield a BIC, IR mechanism, we need to show that for all $j \neq \ell$:

$$t_j(X_j) - p_j \geq t_j(X_\ell) - p_\ell$$
$$\Leftrightarrow t_j(X_j) - t_j(X_\ell) + p_\ell - p_j \geq 0$$

First, consider the case that $j > \ell$. Then:
\begin{align*}
t_j(X_j) - t_j(X_\ell) + p_\ell - p_j &= t_j(X_j) - t_j(X_\ell) - \sum_{z = \ell+1}^j (t_z(X_z) - t_z(X_{z-1}))\\
&= t_j(X_j) + \sum_{z=\ell+1}^j t_z(X_{z-1}) - t_j(X_\ell) - \sum_{z = \ell+1}^j t_z(X_z)\\
&= \sum_{z = \ell+1}^j t_z(X_{z-1}) - t_j(X_\ell) - \sum_{z = \ell+1}^{j-1} t_z(X_{z})\\
\end{align*}

Notice now that $\sum_{z = \ell+1}^j t_z(X_{z-1})$ is exactly the welfare of the matching that gives $X_{z-1}$ to $t_z$ for all $z \in \{\ell+1,\ldots,j\}$. Similarly, $t_j(X_\ell) + \sum_{z = \ell+1}^{j-1} t_z(X_{z-1})$ is exactly the welfare of the matching that gives $X_{z}$ to $t_z$ for all $z \in \{\ell+1,\ldots,j-1\}$ and gives $X_\ell$ to $t_j$. In other words, both sums represent the welfare of a matching between allocations in $\{X_{\ell},\ldots,X_{j-1}\}$ and types in $\{t_{\ell+1},\ldots,t_j\}$. Furthermore, compatibility guarantees that the first sum is larger. Therefore, this term is positive, and $t_j$ cannot gain by misreporting any $t_\ell$ for $\ell < j$ (including $\ell = 0$). 

The case of $j < \ell$ is nearly identical, but included below for completeness:

\begin{align*}
t_j(X_j) - t_j(X_\ell) + p_\ell - p_j &= t_j(X_j) - t_j(X_\ell) + \sum_{z = j+1}^\ell (t_z(X_z) - t_z(X_{z-1}))\\
&= t_j(X_j) + \sum_{z=j+1}^\ell t_z(X_{z}) - t_j(X_\ell) - \sum_{z = j+1}^\ell t_z(X_{z-1})\\
&= \sum_{z = j}^\ell t_z(X_{z}) - t_j(X_\ell) - \sum_{z = j+1}^{\ell} t_z(X_{z-1})\\
\end{align*}

Again, notice that $\sum_{z = j}^\ell t_z(X_z)$ is exactly the welfare of the matching that gives $X_z$ to $t_z$ for all $z \in \{j,\ldots,\ell\}$. Similarly, $t_j(X_\ell) + \sum_{z = j+1}^\ell t_z(X_{z-1})$ is exactly the welfare of the matching that gives $X_{z-1}$ to $t_z$ for all $z \in \{j+1,\ldots,\ell\}$, and gives $X_\ell$ to $t_j$. In other words, both sums represent the welfare of a matching between allocations in $\{X_j,\ldots,X_\ell\}$ and types in $\{t_j,\ldots,t_{\ell}\}$. Furthermore, cyclic monotonicity guarantees that the first sum is larger. Therefore, this term is positive, and $t_j$ cannot gain by misreporting any $t_\ell$ for $\ell > j$. Putting both cases together proves that these prices yield a BIC, IR mechanism, and therefore the bound in Equation~\eqref{eq:revenue} is attained.

\end{prevproof}

\subsection{Relating SADP to MDMDP}\label{app:SADPtoMDMDP2}
Proposition~\ref{prop:exactrevenue} tells us how, when given an allocation rule that is compatible with the type space, to find prices that achieve the maximum revenue. However, if our goal is to maximize revenue over all feasible allocation rules, it does not tell us what the optimal allocation rule is. Now, let us take a closer look at the bound in Equation~\eqref{eq:revenue}. Each allocation $X_\ell$ is only ever evaluated by two types: $t_\ell$ and $t_{\ell+1}$. So to maximize revenue, a tempting approach is to choose $X^*_\ell$ to be the allocation that maximizes $(\sum_{j \geq \ell} q_j) t_\ell(X_\ell)- (\sum_{j \geq \ell+1}q_j)t_{\ell+1}(X_\ell)$, and hope that that $\vec{X}^*$ and $\vec{t}$ are compatible. Because we chose the $X^*_\ell$ to maximize the upper bound of Equation~\eqref{eq:revenue}, the optimal {obtainable revenue for any allocation rule can not possibly }exceed the upper bound of Equation~\eqref{eq:revenue} when evaluated at the $X^*_\ell$ (by Proposition~\ref{prop:upperboundrevenue}), and if $\vec{X}^*$ is compatible with $\vec{t}$, then Proposition~\ref{prop:exactrevenue} tells us that this revenue is in fact attainable. Keeping these results in mind, we now begin drawing connections to SADP. For simplicity of notation in the definitions below, we let $f_{k+1}(\cdot)$ be the $0$ function.

\begin{definition} ($C$-compatible) We say that $(f_1,\ldots,f_k)$ are \emph{$C$-compatible} iff there exist multipliers $1 = Q_1 < Q_2 < \ldots < Q_k $, of bit complexity  at most $C$, and allocations $(X^*_1,\ldots,X^*_k)$ such that $X^*_\ell$ maximizes $f_\ell(\cdot) - f_{\ell+1}(\cdot)$ for all $\ell$ and $(X^*_1,\ldots,X^*_k)$ is compatible with with $(Q_1f_1,\ldots,Q_kf_k)$.
\end{definition}
\begin{observation}\label{obs:useful} Let $1 = Q_1 < Q_2 < \ldots < Q_k$. Then
$$Rev((Q_1f_1,\ldots,Q_kf_k),(1/Q_1-1/Q_2,1/Q_2 - 1/Q_3,\ldots,1/Q_k),\vec{X}) \leq \sum_{\ell=1}^{k}f_\ell(X_\ell) - f_{\ell+1}(X_\ell)$$
This bound is tight when $\vec{X}$ is compatible with $(Q_1f_1,\ldots,Q_kf_k)$.
 \end{observation}
\begin{proof}
Plug in to Propositions~\ref{prop:upperboundrevenue} and~\ref{prop:exactrevenue}.
\end{proof}

\begin{definition} ($D$-balanced) For a list of functions $(f_1,\ldots,f_k)$, let $X^*_\ell$ denote the allocation that maximizes $f_\ell(\cdot) - f_{\ell+1}(\cdot)$ for all $\ell\in[k]$. We say that $(f_1,\ldots,f_k)$ are \emph{$D$-balanced} if $f_k(X^*_k) \leq D(f_\ell(X^*_\ell) - f_{\ell+1}(X^*_\ell))$ for all $\ell\in[k-1]$.
\end{definition}

\begin{proposition}\label{prop:almost}
For a $C$-compatible and $D$-balanced list of functions $(f_1,\ldots,f_k)$, let $X^*_\ell$ denote the allocation that maximizes $f_\ell(\cdot) - f_{\ell+1}(\cdot)$ for all $\ell\in[k]$ and let $1 = Q_1 < Q_2 < \ldots < Q_k$ be multipliers such that $(X^{*}_1,\ldots,X^{*}_k)$ is compatible with $(Q_{1}f_{1},\ldots,Q_{k}f_{k})$. If $(X_{1},\ldots,X_{k})$ are such that \begin{align*}&Rev((Q_1f_1,\ldots,Q_kf_k),(1/Q_1-1/Q_2,1/Q_2 - 1/Q_3,\ldots,1/Q_k),\vec{X})\\ &~~~~~~~~~~~~~~~~\geq\alpha Rev((Q_1f_1,\ldots,Q_kf_k),(1/Q_1-1/Q_2,1/Q_2 - 1/Q_3,\ldots,1/Q_k),\vec{X}^*),\end{align*}then at least one of $\{X_1,\ldots,X_k\}$ is an $\left(\alpha - \frac{(1-\alpha)D}{k-1}\right)$-approximation\footnote{$\alpha - \frac{(1-\alpha)D}{k-1}=1$ when $\alpha=1$.} to the SADP instance $(f_1,\ldots,f_k)$.
\end{proposition}

\begin{prevproof}{Proposition}{prop:almost}
Using Observation~\ref{obs:useful}, we obtain the following chain of inequalities:

\begin{align*}
\sum_{\ell = 1}^k f_{\ell}(X_\ell) - f_{\ell+1}(X_\ell) &\geq Rev((Q_1f_1,\ldots,Q_kf_k),(1/Q_1-1/Q_2,1/Q_2 - 1/Q_3,\ldots,1/Q_k),\vec{X})\\
&\geq \alpha Rev((Q_1f_1,\ldots,Q_kf_k),(1/Q_1-1/Q_2,1/Q_2 - 1/Q_3,\ldots,1/Q_k),\vec{X}^*)\\
&= \alpha \sum_{\ell = 1}^k (f_{\ell}(X^*_\ell) - f_{\ell+1}(X^*_\ell))
\end{align*}

Rearranging, we can get:
\begin{align*}
\sum_{\ell=1}^{k-1} f_\ell(X_\ell) - f_{\ell+1}(X_\ell) &\geq \alpha \left(\sum_{\ell = 1}^{k-1} f_{\ell}(X^*_\ell) - f_{\ell+1}(X^*_\ell)\right) + \alpha f_k(X^*_\ell) - f_k(X_\ell)\\
&\geq \alpha \left(\sum_{\ell = 1}^{k-1} f_{\ell}(X^*_\ell) - f_{\ell+1}(X^*_\ell)\right) - (1-\alpha)f_k(X^*_\ell)
\end{align*}

Now, because $(f_1,\ldots,f_k)$ is $D$-balanced, we have $f_k(X^*_k) \leq D(f_\ell(X^*_\ell) - f_{\ell+1}(X^*_\ell))$ for all $\ell$, and therefore $f_k(X^*_k) \leq \frac{D}{k-1}\left(\sum_{\ell = 1}^{k-1} f_{\ell}(X^*_\ell) - f_{\ell+1}(X^*_\ell)\right)$. Using this, we can again rewrite to obtain:

\begin{equation}\label{eq:last}\sum_{\ell=1}^{k-1} f_\ell(X_\ell) - f_{\ell+1}(X_\ell) \geq \left(\alpha - \frac{(1-\alpha)D}{k-1}\right) \left(\sum_{\ell = 1}^{k-1} f_{\ell}(X^*_\ell) - f_{\ell+1}(X^*_\ell)\right)
\end{equation}

As choosing the null allocation is always allowed, $f_{\ell}(X^*_\ell) - f_{\ell+1}(X^*_\ell)\geq 0$ for all $\ell\in[k]$.  Now it is easy to see that in order for Equation~\eqref{eq:last} to hold, there must be at least one $\ell$ with $f_\ell(X_\ell) - f_{\ell+1}(X_\ell) \geq \left(\alpha - \frac{(1-\alpha)D}{k-1}\right)(f_\ell(X^*_\ell) - f_{\ell+1}(X^*_\ell))$. Such an $X_\ell$ is clearly an $\left(\alpha - \frac{(1-\alpha)D}{k-1}\right)$-approximation to the desired SADP instance.\end{prevproof}

With Proposition~\ref{prop:almost}, we are now ready to state our reduction from SADP to MDMDP.
\begin{theorem}\label{thm:reduction}
Let $A$ be an $\alpha$-approximation algorithm for MDMDP($\mathcal{F}$,$\mathcal{V}^*$).{\footnote{{Note that if $\mathcal{V}$ is closed under addition and scalar multiplication then $\mathcal{V}^* = \mathcal{V}$.}}} {Then an approximate solution to any $C$-compatible instance $(f_{1},\ldots,f_{k})$ of  SADP($\mathcal{F}$,$\mathcal{V}$)} can be found in polynomial time plus one black-box call to $A$. The solution has the following properties:
\begin{enumerate}
\item (Quality) If $(f_1,\ldots,f_k)$ is $D$-balanced, then the solution obtains an $\left(\alpha - \frac{(1-\alpha)D}{k-1}\right)$-approximation.
\item (Bit Complexity) If $b$ is an upper bound on the bit complexity of $f_j(X)$, then $b+C$ is an upper bound on the bit complexity of $t_j(X)$ for any $t_j$ input to $A$ and any $X \in \mathcal{F}$, and every probability input to $A$ has bit complexity $C$.
\end{enumerate}

\end{theorem}

Theorem~\ref{thm:reduction} shows that there is a complete circle of reductions from MDMDP to SADP and back. {In fact, the circle of exact reductions is complete even if we restrict SADP to instances with $k=2$.} The reduction from SADP to MDMDP is not quite approximation preserving, but it is strong enough for us to show hardness of approximation for submodular bidders in the following section.

\subsection{Approximation Hardness for Submodular Bidders}\label{sec:NPhardSubmodular}
We begin this section by defining submodularity. There are several known equivalent definitions for submodularity. The one that will be most convenient for us is known as the property of diminishing returns.
\begin{definition}(Submodular Function) A function $f:2^S \rightarrow \mathbb{R}$ is said to be submodular if for all $X \subset Y \subset S$, and all $x \notin Y$, we have:

$$f(X \cup \{x\}) - f(X) \geq f(Y \cup \{x\}) - f(Y)$$
\end{definition}

\begin{definition}(Value Oracle) A \emph{value oracle} is a black box that takes as input a set $S$ and outputs $f(S)$.
\end{definition}

\begin{definition}(Demand Oracle) A \emph{demand oracle} is a black box that takes as input a vector of prices $p_1,\ldots,p_n$ and outputs $\argmax_{S \subseteq [n]} \{f(S) - \sum_{i \in S} p_i\}$. 
\end{definition}

\begin{definition}(Explicit Representation) An \emph{explicit representation} of $f$ is an explicit description of a turing machine that computes $f$.
\end{definition}

We continue now by providing our hard SADP instance. Let $S_1,\ldots,S_k$ be any subsets of $[n]$ with $|S_i| \leq |S_j|$ for all $i \leq j$. Define $f_i:2^{[n]} \rightarrow \mathbb{R}$ as the following:

$$f_i(S) = \left\{
\begin{array}{lr}
       2n|S| - |S|^2  &: S \notin \{S_1,\ldots,S_{i}\}\\
       2n|S| - |S|^2 - \mattnote{1} + \frac{j}{\mattnote{2}i} & : S = S_j\ \forall j\in [i]
            \end{array}
            \right.
$$

\begin{lemma}\label{lem:submodular}
Each $f_i$ defined above is a monotone submodular function.
\end{lemma}

\begin{prevproof}{Lemma}{lem:submodular}
Consider any $X \subset Y$, and any $x \notin Y$. Then:
\begin{align*}
f_i(X \cup \{x\}) - f_i(X) &\geq 2n(|X|+1) - (|X|+1)^2 - 1 - 2n|X| + |X|^2 &= 2n - 2|X|\mattnote{ -2}\\
f_i(Y \cup \{x\}) - f_i(Y) &\leq 2n(|Y|+1) - (|Y|+1)^2 - 2n|Y| + |Y|^2  \mattnote{+1} &= 2n - 2|Y|
\end{align*}

As $X \subset Y$, we clearly have $|X| < |Y|$, and therefore $f_i(X \cup \{x\}) \geq f_i(Y \cup \{x\})$, so $f_i$ is submodular. Furthermore, as $|X| \leq n$ always, $f_i$ is monotone.
\end{prevproof}

\begin{proposition}\label{prop:submoDifference}
For any $i\in[n-1]$, $S_{i+1}$ maximizes $f_{i}(\cdot)-f_{i+1}(\cdot)$.
\end{proposition}

\begin{prevproof}{Proposition}{prop:submoDifference}
For any $S\notin\{S_{1},\ldots,S_{i+1}\}$, $f_{i}(S)-f_{i+1}(S)=0$. For any $S_{j}\in\{S_{1},\ldots,S_{i}\}$, $f_{i}(S_{j})-f_{i+1}(S_{j})=\frac{j}{\mattnote{2}i}-\frac{j}{\mattnote{2}i+\mattnote{2}}<\frac{1}{2}$. But $f_{i}(S_{i+1})-f_{i+1}(S_{i+1})=\mattnote{\frac{1}{2}}$, so $S_{i+1}$ maximizes $f_{i}(\cdot)-f_{i+1}(\cdot)$.
\end{prevproof}

The idea is in order to solve the SADP instance $(f_1,\ldots,f_k)$ approximately, we need to find one of the $S_i$s with non-neglible probability. Depending on how each $f_i$ is given as input, this may either be easy (if it is given in the form specified above), impossible (if it is given as a value or demand oracle), or {computationally} hard (if it is given via an arbitrary explicit representation). Impossibility for value and demand oracles is straight-forward, and proved next. The {computational} hardness result requires only a little extra work.

\begin{lemma}\label{lem:impossible} Let $\mathcal{S} = \{S_1,\ldots,S_{a}\}$ be subsets chosen uniformly at random {(without replacement)} from all subsets of $[n]$, then listed in increasing order based on their size. Then no algorithm exists that can guarantee an output of $S \in \mathcal{S}$ with probability $\frac{a(c+1)}{2^n-c}$ given only $c$ value or demand oracle queries of $f_{1},\ldots,f_{k}$. 
\end{lemma}

\begin{prevproof}{Lemma}{lem:impossible}
For any sequence of $c$ value queries, the probability that none of them guessed a set $S\in\mathcal{S}$ is at least $1-ac/(2^{n}-c).$ For a deterministic algorithm using $c$ value queries, if it has never successfully queried an element in $\mathcal{S}$ during the execution, then the probability of success is no more than $a/(2^{n}-c)$. Using union bound, we know that the probability of success for any deterministic algorithm is at most $a(c+1)/(2^{n}-c)$. By Yao's Lemma, this also holds for randomized algorithms.

A similar argument shows that the same holds for demand queries as well. First, let's understand how demand queries for $f_0$ ($f_0(S) = 2n|S| - |S|^2$) and $f_i$ differ. Consider any prices $p_1,\ldots,p_n$, and let $S = \argmax\{f_0(S) - \sum_{i \in S} p_i\}$. If $S \notin \mathcal{S}$, then because $f_a(S) = f_0(S)$ and $f_a(S') \leq f_0(S')$ for all $S'$, we also have $S = \argmax\{f_a(S) - \sum_{i \in S} p_i\}$. If other words, if the prices we query are such that a demand oracle for $f_0$ would return some $S \notin \mathcal{S}$, we have learned nothing about $\mathcal{S}$ other than that it does not contain $S$. From here we can make the exact same argument in the previous paragraph replacing the phrase ``query some $S \in \mathcal{S}$'' with ``query some prices $p_1,\ldots,p_n$ such that a demand oracle for $f_0$ would return some $S \in \mathcal{S}$.''
\end{prevproof}

\begin{corollary}\label{cor:SADPhard}
For any constants $c,d$, no (possibly randomized) algorithm can guarantee a $\frac{1}{n^c}$-approximation to SADP($2^{[n]}$,monotone submodular functions) in time $(kn)^d$, as long as $k \leq n^{c'}$ for some constant $c'$, if the functions are given via a value or demand oracle.
\end{corollary}

\begin{prevproof}{Corollary}{cor:SADPhard}
Consider the input $f_1,\ldots,f_k$ chosen according to Lemma~\ref{lem:impossible}. Therefore by Proposition~\ref{prop:submoDifference}, in any $\frac{1}{n^c}$-approximation to this SADP instance necessarily finds some $S_i$ with probability at least $\frac{1}{n^c}$. Lemma~\ref{lem:impossible} exactly says that this cannot be done without at least $2^n/k >> (kn)^d$ queries.
\end{prevproof}

Corollary~\ref{cor:SADPhard} shows that SADP is impossible to approximate within any polynomial factor for submodular functions given via a value or demand oracle. We continue now by showing computational hardness when the functions are given via an explicit representation. First, let $g$ be any correspondence between subsets of $[n]$ and integers in $[2^n]$ so that $|S| > |S'| \Rightarrow g(S) > g(S')$.\footnote{Such a $g$ can be implemented efficiently by ordering subsets of $[n]$ first by their size, and then lexicographically, and letting $g(S)$ be the rank of $S$ in this ordering. The rank of $S$ in this ordering can be computed by easy combinatorics plus a simple dynamic program.} Next, let $\mathcal{P}$ denote any NP-hard problem with verifier $V$, and let $p$ be a constant so that any ``yes'' instance of $\mathcal{P}$ of size $m$ has a bitstring witness of length $m^p$. For a given $n, k$, and specific instance $P \in \mathcal{P}$ of size $(n-\log k)^{1/p}$, we define:

$$f^P_i(S) = \left\{
\begin{array}{lr}
       2n|S| - |S|^2 - \mattnote{1} + {\frac{j}{\mattnote{2}i}} &: \frac{(j-1)2^n}{k} < g(S)\leq \frac{j2^n}{k} \leq \frac{i2^n}{k}\text{ and } V\left(P,(g(S) \mod \frac{2^n}{k} )\right) = ``yes'' \\
       2n|S| - |S|^2 & : \text{otherwise}\\
                   \end{array}
            \right.
$$

\begin{lemma}\label{lem:witness} If $P$ is a ``yes'' instance of $\mathcal{P}$, then any $\alpha$-approximation to the SADP instance $f_1,\ldots,f_k$ necessarily finds a witness for $P$ with probability at least $\alpha$.
\end{lemma}

\begin{prevproof}{Lemma}{lem:witness}
It is clear that if $P$ is a ``yes'' instance of $\mathcal{P}$, then $P$ has a witness of length $(n-\log k)$, and therefore there is some $x \leq 2^n/k$ such that $V(P,x) = ``yes.''$ Therefore, we would have $f^P_i(g^{-1}(x+i2^n/k)) = 2n|S| - |S|^2$, but $f^P_{i+1}(g^{-1}(x+i2^n/k)) = 2n|S| - |S|^2 -\mattnote{\frac{1}{2}}$. It is also clear that in order to have $f^P_i(S) - f^P_{i+1}(S) > 0$, we must have $V(P,(g(S) \mod 2^n/k)) = ``yes''$. In other words, if $P$ is a ``yes'' instance of $\mathcal{P}$, any $\alpha$-approximation to the SADP instance $f_1,\ldots,f_k$ must find a witness for $P$ with probability at least $\alpha$. 
\end{prevproof}

\begin{corollary}\label{cor:NPhard}
Unless $NP = RP$, for any constants $c,d$ no (possibly randomized) algorithm can guarantee a $\frac{1}{n^c}$-approximation to SADP($2^{[n]}$,monotone submodular functions) in time $(kn)^d$, as long as $k \leq n^{c'}$ for some constant $c'$, {even if} the functions are given explicitly.
\end{corollary}

\begin{prevproof}{Corollary}{cor:NPhard}
If we could obtain a $\frac{1}{n^c}$-approximation to SADP($2^{[n]}$,monotone submodular functions) in time $(kn)^d \leq n^{c'd+d}$, then for any ``yes'' instance $P \in \mathcal{P}$, we could find a witness with probability at least $\frac{1}{n^c}$ in time $n^{c'd+d}$ by Lemma~\ref{lem:witness}. By running $n^c$ independent trials, we could amplify this probability to $1/e$ in time $n^{c'd+d+c}$. So we would have an $RP$ algorithm for $\mathcal{P}$. 
\end{prevproof}

We combine Corollaries~\ref{cor:SADPhard} and~\ref{cor:NPhard} below into one theorem.

\begin{theorem}\label{thm:SADPhard}
As long as $k \leq n^{c'}$ for some constant $c'$, the problem SADP($2^{[n]}$,monotone submodular functions) is:
\begin{enumerate}
\item Impossible to approximate within any $1/\poly(n)$-factor with only $\poly(k,n)$ value oracle queries.
\item Impossible to approximate within any $1/\poly(n)$-factor with only $\poly(k,n)$ demand oracle queries.
\item Impossible to approximate within any $1/\poly(n)$-factor given explicit access to the input functions in time $\poly(k,n)$, unless $NP = RP$.
\end{enumerate}
\end{theorem}

Finally, we conclude this section by showing that the classes which we have just shown to be hard for SADP can be used via the reduction of Theorem~\ref{thm:reduction}.

\begin{lemma}\label{lem:compatible}
All SADP instances defined in this section are {$2k(\mattnote{1 + }\log k + \log n)$}-compatible and $\mattnote{2}n^2$-balanced. 
\end{lemma}

\begin{prevproof}{Lemma}{lem:compatible}
That each instance is $\mattnote{2}n^2$-balanced is easy to see: $f_k([n]) = n^2$, and {for all $i$ we have} $\max_S \{f_i(S) - f_{i+1}(S)\} = \mattnote{\frac{1}{2}}$. To see that each instance is $C$-compatible, set $Q_i = (\mattnote{2}kn)^{2i-2}$. Then it is easy to see that the maximum welfare obtainable by any matching that only uses types $Q_1f_1$ through $Q_if_i$ obtains welfare at most $n^2iQ_i$. In addition, the difference in value of $Q_{i+1}f_{i+1}$ between two outcomes is at least {$Q_{i+1}/(\mattnote{2}i+\mattnote{2}) \geq Q_{i+1}/(\mattnote{2}k)$} if it is non-zero. Therefore, the minimum non-zero difference between the value of $Q_{i+1}f_{i+1}$ for two allocations is larger than the maximum possible difference in welfare of all types $Q_1f_1$ through $Q_if_i$ for two matchings of allocations. As this holds for all $i$, this means that the max-weight matching of any set of allocations $S$ to types $\{Q_if_i,\ldots,Q_j f_{j}\}$ will necessarily match $Q_j f_{j}$ to its favorite allocation in $S$, then $Q_{j-1}f_{j-1}$ to its favorite of what remains, etc.

So let $(S_1,\ldots,S_k)$ be any sets such that $f_i(S_{i+1}) - f_i(S_i) = \mattnote{\frac{1}{2}}$. It is not hard to verify that $|S_{1}|\leq|S_{2}|\leq \ldots \leq |S_{k}|$, and $f_{i}(S_{j})$ is monotonically increasing in $j$ for any $i$. The reasoning above immediately yields that the max-weight matching of allocations $(S_i,\ldots,S_j)$ to types $(Q_if_i,\ldots,Q_j f_j)$ necessarily awards $S_\ell$ to $Q_\ell f_\ell$ for all $\ell$, as $f_{i}$, so $(S_1,\ldots,S_k)$ is cyclic monotone with respect to $(Q_1S_1,\ldots,Q_kS_k)$. Furthermore, the same reasoning immediately yields that the max-weight matching of allocations $(S_i,\ldots,S_{j-1})$ to types $(Q_{i+1}f_{i+1},\ldots,Q_jf_j)$ awards $S_{\ell}$ to type $Q_{\ell+1}f_{\ell+1}$ for all $\ell$, so $(S_1,\ldots,S_k)$ is also compatible with $(Q_1f_1,\ldots,Q_kf_k)$. It is clear that all $Q_i$ {are integers less than} $(\mattnote{4}k^2n^{2})^{k-1} \leq 2^{2k(\mattnote{1+}\log k + \log n)}$, so $(f_1,\ldots,f_k)$ are $2k(\mattnote{1+}\log k+\log n)$-compatible.
\end{prevproof}

\begin{corollary}\label{cor:MDMDPhard}
The problem MDMDP($2^{[n]}$,monotone submodular functions) (for $k = |T_1| = \poly(n)$) is:
\begin{enumerate}
\item Impossible to approximate within any $1/\poly(n)$-factor with only $\poly(k,n)$ value oracle queries.
\item Impossible to approximate within any $1/\poly(n)$-factor with only $\poly(k,n)$ demand oracle queries.
\item Impossible to approximate within any $1/\poly(n)$-factor given explicit access to the input functions in time $\poly(k,n)$, unless $NP = RP$.
\end{enumerate}
\end{corollary}

\begin{prevproof}{Corollary}{cor:MDMDPhard}
For any constant $c$, let $k=n^{c+2}+1$, if we can find an $2/n^{c}$-approximate solution to MDMDP in polynomial time, we can find an $1/n^{c}$-approximate solution to SADP by Theorem~\ref{thm:reduction}, which would contradict Theorem~\ref{thm:SADPhard}.
\end{prevproof}

\notshow{\begin{proof}
An algorithm attempting to guess some $S \in \mathcal{S}$ starts with the belief that each $S \subset [n]$ is in $\mathcal{S}$ with probability $a/2^n$. Each query allows the algorithm to update this belief. At the end of $c$ queries, the best thing the algorithm can do is guess the set with the highest probability of being in $\mathcal{S}$. Our goal is to show that with very high probability, this is very low. 

Let's consider how the beliefs get updated from query to query for value queries. Assume that we currently believe that $\mathcal{S}$ is equally likely to be any subset of size $a$ of some $X \subseteq 2^{[n]}$, and then we make a value query for the set $S$. If $f_i(S) = 2n|S| - |S|^2$, then all we have learned is that $S \notin \{S_{1},\ldots, S_{i} \}$, and our belief is updated to be that $\mathcal{S}$ is equally likely to be any subset of size $a$ of $X - \{S\}$. 

It is easy to see now that if we only query elements not in $\mathcal{S}$, then our belief at the end of $c$ queries will be that $\mathcal{S}$ is equally likely to be any subset of size $a$ of some $X$ with $|X| = 2^n - c$. It is also easy to see that in order for us to query an element in $\mathcal{S}$, we must have at some point guessed an $S \in \mathcal{S}$ while believing that $\mathcal{S}$ was equally likely to be any subset of size $c$ of some $X$ with $|X| \geq 2^n - c$. The probability that this happens on our $i^{th}$ query is $\frac{a}{2^n-i}$. By the union bound, the probability that this happens on any of the first $c$ queries is at most $\frac{ac}{2^n-c}$. So with probability at least $1- \frac{ac}{2^n-c}$, our queries are such that our best guess at the end succeeds with probability at most $\frac{a}{2^n-c}$. Even if we succeed with probability $1$ the rest of the time, this means that the maximum possible probability of success is $\frac{a(c+1)}{2^n-c}$.

A similar argument shows that the same holds for demand queries as well. First, let's understand how demand queries for $f_0$ ($f_0(S) = 2n|S| - |S|^2$) and $f_a$ differ. Consider any prices $p_1,\ldots,p_n$, and let $S = \argmax\{f_0(S) - \sum_{i \in S} p_i\}$. If $S \notin \mathcal{S}$, then because $f_a(S) = f_0(S)$ and $f_a(S') \leq f_0(S')$ for all $S'$, we also have $S = \argmax\{f_a(S) - \sum_{i \in S} p_i\}$. If other words, if the prices we query are such that a demand oracle for $f_0$ would return some $S \notin \mathcal{S}$, we have learned nothing about $\mathcal{S}$ other than that it does not contain $S$. From here we can make the exact same argument in the previous paragraph replacing the phrase ``query some $S \in \mathcal{S}$'' with ``query some prices $p_1,\ldots,p_n$ such that a demand oracle for $f_0$ would return some $S \in \mathcal{S}$.''
\end{proof}}
\section{General Objectives (Complete)}\label{app:generalappendix}
Here we give a complete description and proof of our reduction for general objectives. We follow the same outline as that of Section~\ref{sec:revenue}: In Section~\ref{sec:LPgeneral}, we write a poly-size linear program that finds the optimal implicit form provided that we have a separation oracle for $F(\mathcal{F},\mathcal{D},\obj)$. In Section~\ref{sec:MDMDPtoSADPgeneral} we show that any poly-time $\alpha$-approximation algorithm for SADP($\mathcal{F}$,$\mathcal{V}$,$\obj$) implies a poly-time weird separation oracle for $\alpha F(\mathcal{F},\mathcal{D},\obj)$, and therefore a poly-time $\alpha$-approximation algorithm for MDMDP($\mathcal{F}$, $\mathcal{V}$, $\obj$).

\subsection{Linear Programming Formulation}\label{sec:LPgeneral}
We now show how to write a poly-size linear program to find the implicit form of a mechanism that solves the MDMDP. Similar to revenue, the idea is that we will search over all feasible, BIC, IR implicit forms for the one that maximizes $\obj$ in expectation. We first show that $F(\mathcal{F},\mathcal{D},\obj)$ is always a convex region, then state the linear program and prove that it solves MDMDP.

\begin{lemma}\label{lem:convexregion} $F(\mathcal{F},\mathcal{D},\obj)$ is a convex region.
\end{lemma}

\begin{prevproof}{Lemma}{lem:convexregion}
Consider any two feasible implicit forms $\vec{\pi}^I_1$ and $\vec{\pi}^I_2$. Then there exist corresponding mechanisms $M_1 = (A_1,P_1)$ and $M_2 = (A_2,P_2)$ with $\vec{\pi}(M_i) = \vec{\pi}_i$, $\vec{P}(M_i) = \vec{P}_i$ and $\pi_\obj(M_i) \geq (\pi_\obj)_i$. So consider the mechanism that runs $M_1$ with probabilty $c$ and $M_2$ with probability $1-c$. This mechanism clearly has allocation and price components $\vec{\pi}(cM_1 + (1-c)M_2) = c\vec{\pi}_1 + (1-c)\vec{\pi}_2$ and $\vec{P}(cM_1 + (1-c)M_2) = c\vec{P}_1 + (1-c)\vec{P}_2$. Also, because $\obj$ is concave, we must have $\pi_\obj(cM_1 + (1-c)M_2) \geq c\pi_{\obj}(M_{1})+(1-c)\pi_{\obj}(M_{2})\geq c(\pi_\obj)_1 + (1-c)(\pi_\obj)_2$. Therefore, the mechanism $cM_1 + (1-c)M_2$ bears witness that the implicit form $c\vec{\pi}^I_1 + (1-c)\vec{\pi}^I_2$ is also feasible, and $F(\mathcal{F},\mathcal{D},\obj)$ is necessarily convex.
\end{prevproof}

\begin{figure}[ht]
\colorbox{MyGray}{
\begin{minipage}{\textwidth} {
\noindent\textbf{Variables:}
\begin{itemize}
\item $\pi_i(t_i,t'_i)$, for all bidders $i$ and types $t_i,t'_i \in T_i$, denoting the expected value obtained by bidder $i$ when their true type is $t_i$ but they report $t'_i$ instead.
\item $P_i(t_i)$, for all bidders $i$ and types $t_i \in T_i$, denoting the expected price paid by bidder $i$ when they report type $t_i$.
\item $\pi_\obj$, denoting the expected value of $\obj$.
\end{itemize}
\textbf{Constraints:}
\begin{itemize}
\item $\pi_i(t_i,t_i) - P_i(t_i) \geq \pi_i(t_i,t'_i) - P_i(t'_i)$, for all bidders $i$, and types $t_i,t'_i \in T_i$, guaranteeing that the implicit form {$(\pi_\obj,\vec{\pi},\vec{P})$} is BIC.
\item $\pi_i(t_i,t_i) - P_i(t_i) \geq 0$, for all bidders $i$, and types $t_i \in T_i$, guaranteeing that the implicit form {$(\pi_\obj,\vec{\pi},\vec{P})$} is individually rational.
\item \begin{itemize}
		\item General $\obj$: $(\pi_\obj,\vec{\pi},\vec{P}) \in F(\mathcal{F},\mathcal{D},\obj)$, guaranteeing that the 		implicit form $(\pi_\obj,\vec{\pi},\vec{P})$ is feasible.
		\item Allocation-only $\obj$: $(\pi_\obj,\vec{\pi})\in F(\mathcal{F},\mathcal{D},\obj)$, guaranteeing that the 		implicit form $(\pi_\obj,\vec{\pi},\vec{P})$ is feasible.
		\item Price-only $\obj$: $(\pi_\obj,\vec{P})\in F(\mathcal{F},\mathcal{D},\obj)$ and $\vec{\pi} \in F(\mathcal{F},\mathcal{D})$, guaranteeing that the implicit form $(\pi_\obj,\vec{\pi},\vec{P})$ is feasible.
		\end{itemize}
\end{itemize}
\textbf{Maximizing:}
\begin{itemize}
\item $\pi_\obj$, the expected value of $\obj$ when played truthfully by bidders sampled from $\mathcal{D}$.\\
\end{itemize}}
\end{minipage}}
\caption{A linear programming formulation for the general MDMDP.}
\label{fig:LPMDMDgeneral}
\end{figure}

\begin{observation}\label{cor:LPsolvesMDMDPgeneral}
Any $\alpha$-approximate solution to the linear program of Figure~\ref{fig:LPMDMDgeneral} corresponds to a feasible, BIC, IR implicit form whose expected value of $\obj$ is at least an $\alpha$-fraction of the optimal obtainable expected value of $\obj$ by a feasible, BIC, IR mechanism. 
\end{observation}

\begin{proof}
Same as Observation~\ref{obs:LPsolvesMDMDP}.
\end{proof}

\begin{corollary}\label{cor:LPgeneral} The program in Figure~\ref{fig:LPMDMDgeneral} is a linear program with $1 + \sum_{i\in[m]} (|T_{i}|^2 + |T_{i}|)$ variables. \mattnote{In addition, if $b$ is an upper bound on the bit complexity of the  extreme points of $F(\mathcal{F},\mathcal{D},\obj)$, then }with black-box access to a weird separation oracle, $WSO$, for  $\alpha F(\mathcal{F},\mathcal{D},\obj)$, the implicit form of an $\alpha$-approximate solution to MDMDP can be found in time polynomial in $\sum_{i\in[m]} |T_{i}|$, $b$, and the runtime of $WSO$ on inputs with bit complexity polynomial in $\sum_{i\in[m]} |T_{i}|$, $b$.
\end{corollary}

\begin{proof}
It is clear that the variable count is correct, and that the BIC and IR constraints are linear. By Lemma~\ref{lem:convexregion}, $F(\mathcal{F},\mathcal{D},\obj)$ is a convex region. \mattnote{If in addition $b$ upper bounds the bit complexity of the extreme points of $F(\mathcal{F},\mathcal{D},\obj)$, then $F(\mathcal{F},\mathcal{D},\obj)$ must indeed be a polytope (as there can only be finitely many extreme points).} It is also clear that the objective function is linear. 

That an $\alpha$-approximate solution can be found using $WSO$ is a consequence of Corollary~\ref{cor:CaiDW1}. Denote by $\vec{\pi}_1$ the optimal solution to the LP in Figure~\ref{fig:LPMDMDgeneral}, $\vec{\pi}_2$ the optimal solution after replacing $F(\mathcal{F},\mathcal{D},\obj)$ with $\alpha F(\mathcal{F},\mathcal{D},\obj)$ and $\vec{\pi}_3$ the optimal solution after replacing $F(\mathcal{F},\mathcal{D},\obj)$ with $WSO$ for $\alpha F(\mathcal{F},\mathcal{D},\obj)$. Denote by $c_1,c_2$, and $c_3$ the $\pi_\obj$ component of each. It is easy to see that $\alpha \vec{\pi}_1$ is a feasible solution to the LP using $\alpha F(\mathcal{F},\mathcal{D},\obj)$: $\alpha \vec{\pi}_1$ is in $\alpha F(\mathcal{F},\mathcal{D},\obj)$ by definition, and the BIC and IR constraints are invariant to scaling. Furthermore, it is clear that the objective component of $\alpha \vec{\pi}_1$ is exactly $\alpha c_1$. We may therefore conclude that $c_2 \geq \alpha c_1$. As $WSO$ is a weird separation oracle for $\alpha F(\mathcal{F},\mathcal{D},\obj)$, Corollary~\ref{cor:CaiDW1} guarantees that $c_3 \geq c_2 \geq \alpha c_1$, and that $\vec{\pi}_3$ can be found using the same parameters as a valid separation orcale for $\alpha F(\mathcal{F},\mathcal{D},\obj)$. So using Theorem~\ref{thm:ellipsoid} and replacing the running time of $SO$ with the running time of $WSO$, we obtain the desired running time. Finally, by Observation~\ref{cor:LPsolvesMDMDPgeneral}, $\vec{\pi}_3$ corresponds to the implicit form of an $\alpha$-approximate solution to MDMDP.
\end{proof}

\subsection{Reduction from MDMDP to SADP}\label{sec:MDMDPtoSADPgeneral}
Based on Corollary~\ref{cor:LPgeneral}, the only obstacle to solving the MDMDP is obtaining a separation oracle for $F(\mathcal{F},\mathcal{D},\mathcal{O})$ (or ``weird'' separation oracle for $\alpha F(\mathcal{F},\mathcal{D},\obj)$). In this section, we again use Theorem~\ref{thm:CaiDW1} to obtain a weird separation oracle for $\alpha F(\mathcal{F},\mathcal{D},\obj)$ using only black box access to an $\alpha$-approximation algorithm for SADP.

Again, we can apply Theorem~\ref{thm:CaiDW1} as long as we understand what it means to compute $\vec{x} \cdot \vec{w}$ in our setting. Recall that $\vec{x}$ is some implicit form $\vec{\pi}^I$, so the direction $\vec{w}$ has components $w_\obj$, $w_i(t_i,t'_i)$ for all $i,t_i,t'_i$, and $W_i(t_i)$ for all $i, t_i$. We first prove the analogue of Proposition~\ref{prop:virtualwelfare} below, showing that $\vec{\pi}^I \cdot \vec{w}$ can be interpreted as maximizing the expected value of a virtual objective.

\mattnote{In the proof of Proposition~\ref{prop:virtualobjective}, as well as the remainder of this section, we will be considering directions that may assign a negative multiplier $w_\obj$ to the objective component of an implicit form. However, we would really like not to have to consider instances of SADP that also have a negative multiplier in front of the objective. Fortunately, this can be achieved without much work. For the remainder of the section, we will denote by the ``corresponding implicit form'' of a mechanism exactly the implicit form of $M$ if $w_\obj > 0$, or the feasible implicit form $(0,\vec{\pi}(M),\vec{P}(M))$ if $w_\obj < 0$. For any feasible implicit form $\vec{\pi}^I_0 = (x,\vec{\pi}(M),\vec{P}(M))$, it is clear that both $(0,\vec{\pi}(M),\vec{P}(M))$ and $(\pi_\obj(M),\vec{\pi}(M),\vec{P}(M))$ are feasible, and that the corresponding implicit form $\vec{\pi}^I$ of $M$ has $\vec{\pi}^I \cdot \vec{w} \geq \vec{\pi}^I_0 \cdot \vec{w}$.}

\begin{proposition}\label{prop:virtualobjective}
Let $\vec{w}$ be a direction in $[-1,1]^{1+\sum_i (|T_i|^2+ |T_i|)}$, and let $\hat{w}_{\obj}=w_{\obj}$ if $w_{\obj}>0$, otherwise $\hat{w}_{\obj}=0$. Define the virtual objective $\obj'_{\vec{w}}$ as:

$$\obj'_{\vec{w}}(\vec{t}',X,P) = \hat{w}_\obj \cdot \obj(\vec{t}',X,P) + \sum_i \sum_{t_i \in T_i} \frac{w_i(t_i,t'_i)}{\Pr[t'_i])}\cdot t_i(X) + \sum_i W_i(t'_i)\mathbb{E}[P_i],$$

Then any mechanism that obtains an $\alpha$-approximation to maximizing the expected value of the virtual objective $\obj'_{\vec{w}}$ corresponds to an implicit form $\vec{\pi}^I$ that is an $\alpha$-approximation to maximizing $\vec{x} \cdot \vec{w}$ over $\vec{x} \in F(\mathcal{F},\mathcal{D},\obj)$.

When $\obj$ is allocation-only, we may remove the pricing rule entirely, getting:

$$\obj'_{\vec{w}}(\vec{t}',X) = \hat{w}_\obj \cdot \obj(\vec{t}',X) + \sum_i \sum_{t_i \in T_i} \frac{w_i(t_i,t'_i)}{\Pr[t'_i])}\cdot t_i(X)$$

When $\obj$ is price-only, and we are looking at the price-only feasible region (for $(\pi_\obj,\vec{P}))$, we may remove the allocation rule entirely, getting:

$$\obj'_{\vec{w}}(P) = \hat{w}_\obj \cdot \obj(P) + \sum_i W_i(t'_i)\mathbb{E}[P_i]$$

When $\obj$ is price-only and we are looking at the feasible region (for $(\vec{\pi})$), we get back the same virtual welfare function used in Proposition~\ref{prop:virtualwelfare} for revenue:

$$\obj'_{\vec{w}}(\vec{t}',X) = \sum_i \sum_{t_i \in T_i} \frac{w_i(t_i,t'_i)}{\Pr[t'_i])}\cdot t_i(X)$$
\end{proposition}
\begin{proof}
Proposition~\ref{prop:virtualwelfare} shows that $\sum_i \sum_{t_i,t'_i} w_i(t_i,t'_i)\cdot \pi_i(t_i,t'_i)$ is exactly the virtual welfare of an allocation with implicit form $\vec{\pi}^I$ computed when the virtual type of $t'_i$ is $\sum_{t_i \in T_i}\frac{w_i(t_i,t'_i)}{\Pr[t'_i])}\cdot t_i(\cdot)$. Furthermore, $\sum_i \sum_{t_i} W_i(t_i) P_i(t_i)$ is clearly just a weighted sum of expected payments (which could be interpreted as virtual revenue). Finally, if $w_\obj \leq 0$, any optimal solution may set $\pi_\obj = 0$, so this term will contribute nothing. If $w_\obj > 0$, then any optimal solution will set $\pi_\obj = \pi_\obj(M)$, and this term will contribute exactly $\pi_\obj(M) \cdot w_\obj$.

Therefore, summing all three terms together we see that the optimal implicit form corresponds to some mechanism $M$ that maximizes the expected value of $\hat{w}_\obj \cdot \obj + \text{ virtual welfare } + \text{ virtual revenue }$, which is exactly $\obj'_{\vec{w}}$. Furthermore, it is easy to see that, in fact, for any mechanism $M$, the corresponding implicit form $\vec{\pi}^I_0$ of $M$ satisfies $\vec{\pi}^I_0 \cdot \vec{w} = $ the expected value of $\obj'_{\vec{w}}$ obtained by $M$. Therefore, if $M$ is an $\alpha$-approximation to maximizing $\obj'_{\vec{w}}$ in expectation, the corresponding implicit form is an $\alpha$-approximation to maximizing $\vec{x} \cdot \vec{w}$ over $F(\mathcal{F},\mathcal{D},\obj)$.

That the simplifications hold for allocation-only and price-only objectives is obvious.
\end{proof}

Recall again that Theorem~\ref{thm:CaiDW1} requires an algorithm $\mathcal{A}$ that takes as input a direction $\vec{w}$ and outputs a $\vec{\pi}^I$ with $ \vec{w} \cdot \vec{\pi}^I \geq \alpha \max_{\vec{x} \in F(\mathcal{F},\mathcal{D},\obj)} \{\vec{w} \cdot \vec{x}\}$. By Proposition~\ref{prop:virtualobjective} above, we know that this is exactly asking for a mechanism that obtains an $\alpha$-fraction of the optimal expected value for $\obj'_{\vec{w}}$ over all feasible mechanisms. Next, we provide an analogue to Corollary~\ref{cor:obv} of Section~\ref{sec:revenue} in Corollary~\ref{cor:obvgeneral} below, which relates the per-profile performance of a mechanism $M$ to $\vec{w} \cdot \vec{\pi}^I(M)$.

\begin{corollary}\label{cor:obvgeneral}
Let $M$ be a mechanism that on profile $(t'_1,\ldots,t'_m)$ chooses a (possibly randomized) allocation $X \in \mathcal{F}$ and (possibly randomized) pricing scheme $P$ such that

$$\obj'_{\vec{w}}(\vec{t}',X,P) \geq \alpha \max_{X' \in \mathcal{F},P'} \{\obj'_{\vec{w}}(\vec{t}',X,P)\}$$

Then the corresponding implicit form  $\vec{\pi}^I$ satisfies:

$$\vec{\pi}^I \cdot \vec{w} \geq \alpha \max_{\vec{x} \in F(\mathcal{F},\mathcal{D},\obj)} \{\vec{x} \cdot \vec{w}\}$$
\end{corollary}

\begin{proof}
Exactly the same as Corollary~\ref{cor:obv}, but using Proposition~\ref{prop:virtualobjective} instead of Proposition~\ref{prop:virtualwelfare}: if a mechanism achieves an $\alpha$-fraction of the optimal value of $\obj'_{\vec{w}}$ on every profile, then it clearly obtains an $\alpha$-fraction of the optimal expected value of $\obj'_{\vec{w}}$.
\end{proof}

With Corollary~\ref{cor:obvgeneral}, we now want to study the problem of maximizing virtual welfare on a given profile. This turns out to be exactly an instance of SADP. 

\begin{proposition}\label{prop:SADPvirtualobjective}
Let $\vec{w} \in [-1,1]^{1+\sum_i (|T_i|^2 + |T_i|)}$ with $w_\obj \geq 0$. Let also $t'_i,t_i \in \mathcal{V}$ for all $i, t'_i,t_i$. Then any $X \in \mathcal{F}, P$ that is an $\alpha$-approximation to the SADP($\mathcal{F}$,$\mathcal{V}$,$\obj$) on input:

\begin{align*}
f_1 &= \sum_i \sum_{t_i \in T_i | w_i(t_i,t'_i) > 0} \frac{w_i(t_i,t'_i)}{\Pr[t'_i])}\cdot t_i(X)\\
f_2 &= \sum_i \sum_{t_i \in T_i | w_i(t_i,t'_i) < 0} \frac{w_i(t_i,t'_i)}{\Pr[t'_i])}\cdot t_i(X)\\
g_i &= t'_i\\
c_i &= W_i(t'_i)\\
c_0 &= \hat{w}_\obj
\end{align*}

is also an $\alpha$-approximation for maximizing $\obj'_{\vec{w}}$. That is:

$$\obj'_{\vec{w}}(\vec{t}',X,P) \geq \alpha\cdot\max_{X' \in \mathcal{F},P} \{\obj'_{\vec{w}}(\vec{t}',X',P)\}$$
\end{proposition}

\begin{proof}
It is clear that $(f_1,f_2,\vec{g},\vec{c},c_0)$ is a valid input to SADP($\mathcal{F}$,$\mathcal{V}$,$\obj$). The remainder of the proof is exactly the same as that of Proposition~\ref{prop:SADPvirtualwelfare}.
\end{proof}

\begin{corollary}\label{cor:importantgeneral}
Let $G$ be any $\alpha$-approximation algorithm for SADP($\mathcal{F}$,$\mathcal{V}$,$\obj$). Let also $M$ be the mechanism that, on profile $(t'_1,\ldots,t'_m)$ chooses the allocation \mattnote{and price rule} $G(f_1,f_2,\vec{g},\vec{c},c_0)$ (where $f_1,f_2,\vec{g},\vec{c}, c_0$ are defined as in Proposition~\ref{prop:SADPvirtualobjective}). Then the corresponding implicit form $\vec{\pi}^I$ satisfies:

$$\vec{\pi}^I \cdot \vec{w} \geq \alpha \max_{\vec{x} \in F(\mathcal{F},\mathcal{D},\obj)} \{\vec{x} \cdot \vec{w}\}$$
\end{corollary}

\begin{proof}
Combine Corollary~\ref{cor:obvgeneral} and Proposition~\ref{prop:SADPvirtualobjective}.
\end{proof}

At this point, we would like to just let $\mathcal{A}$ be the algorithm that takes as input a direction $\vec{w}$ and computes the implicit form prescribed by Corollary~\ref{cor:importantgeneral} . Corollary~\ref{cor:importantgeneral} shows that this algorithm satisfies the hypotheses of Theorem~\ref{thm:CaiDW1}, so we would get a weird separation oracle for $\alpha F(\mathcal{F},\mathcal{D},\obj)$. The same situation as in Section~\ref{sec:revenue} arises, so we must use the same approach of~\cite{CaiDW12b} \mattnote{and define $\mathcal{D}'$ to be a uniform distribution over multi-set over polynomially-many profiles sampled from $\mathcal{D}$}. Formally, our algorithm for MDMDP is stated as Algorithm~\ref{alg:solveMDMDPgeneral} below.

\begin{algorithm}[h!]
        \caption{Algorithm to solve MDMDP($\mathcal{F}$,$\mathcal{V}$,$\obj$):}
       \begin{algorithmic}[1]\label{alg:solveMDMDPgeneral}
        \STATE Input: \mattnote{type spaces $T_i$ for $i \in [m]$,}, distributions $\mathcal{D}_i$ over $T_{i}$, an error $\epsilon$, and black-box access to $G$, a (possibly randomized) $\alpha$-approximation algorithm for SADP($\mathcal{F}$,$\mathcal{V}$,$\obj$).
        \STATE Obtain $\mathcal{D}'$ as the uniform distribution over a multi-set of profiles sampled from $\mathcal{D}$ in the manner described in Section~5.2 of~\cite{CaiDW12b}.
        \STATE Modify $G$ to obtain $G'$ as prescribed in Appendix~G.2 of~\cite{CaiDW13}. $G'$ is a deterministic algorithm.
        \STATE Define $\mathcal{A}$ to, on input $\vec{w} \in [-1,1]^{1 + \sum_i (|T_i|^2+|T_i|)}$, compute the implicit form prescribed in Corollary~\ref{cor:importantgeneral} using $G'$ and $\mathcal{D}'$. 
         \STATE Solve the linear program in Figure~\ref{fig:LPMDMDgeneral} using the weird separation oracle obtained from $\mathcal{A}$ using Theorem~\ref{thm:CaiDW1}. Call the output $(\pi_\obj,\vec{\pi},\vec{P})$, and let $\vec{w}_1,\ldots,\vec{w}_l$ be the directions guaranteed by Theorem~\ref{thm:CaiDW1} to be output by the weird separation oracle.
         \STATE Solve the linear system $\vec{\pi}^I = \sum_{j\in[l] }c_j \mathcal{A}(\vec{w}_j)$ to write $\vec{\pi}^{I}$ as a convex combination of $\mathcal{A}(\vec{w}_j)$.
         \STATE Output the following mechanism: First, randomly select a $\vec{w} \in [-1,1]^{1+\sum_i (|T_i|^2+|T_i|)}$ by choosing $\vec{w}_j$ with probability $c_j$. Then, choose the allocation output by $G'$ on the SADP input corresponding to $\vec{w}$.
             \end{algorithmic}
\end{algorithm}

\medskip We now proceed to prove its correctness. The only big remaining step is to relate the bit complexity of the objective, the probability distribution, and the values to the bit complexity of the extreme points of \mattnote{$F(\mathcal{F},\mathcal{D}',\obj)$}. Recall that Corollary~\ref{cor:LPgeneral} only guarantees that the desired LP can be solved in time polynomial in the bit complexity of the extreme points of \mattnote{$F(\mathcal{F},\mathcal{D}',\obj)$}. However, Theorem~\ref{thm:objective} claims an algorithm whose runtime is polynomial in $b$, which upper bounds the bit complexity of $\obj$, $t_i(X)$, and $Pr[t_i]$ for all $i,t_i \in T_i, X \in \mathcal{F}$. It's not hard to imagine that these two bit complexities can't differ by much, and indeed this is the case. Lemma~\ref{lem:lowbitcorners} below states this formally. 

\begin{lemma}\label{lem:lowbitcorners}
If $b$ is an upper bound on the bit complexity of $\obj$, $Pr[t_i]$, and $t_i(X)$, for all $i,t_i \in T_i, X \in \mathcal{F}$, then all extreme points of $F(\mathcal{F},\mathcal{D}',\obj)$ have bit complexity $\poly(b,\sum_i |T_i|)$. 
\end{lemma}

\begin{prevproof}{Lemma}{lem:lowbitcorners}
By Proposition~\ref{prop:virtualobjective}, we know that all extreme points of\mattnote{ $F(\mathcal{F},\mathcal{D}',\obj)$} correspond to virtual objective maximizers. We also know that virtual objective maximizers simply maximize the virtual objective on every profile. So let's try to understand what the picture looks like for a single profile by figuring out what kinds of allocations can possibly maximize a virtual objective.

For a single profile, recall that the feasible choices combine any randomized allocation in $\Delta(\mathcal{F})$ with any pricing scheme that charges each bidder $i$ some price in $[0,1]$ \mattnote{(recall that this assumption was made w.l.o.g. in the preliminaries of Section~\ref{sec:generalprelim}).} We will represent these constraints as a polytope \mattnote{$F'$ in $|\mathcal{F}| + m + 1$ dimensions using} the following variables: $Q(X)$ for all $X \in \mathcal{F}$ denoting the probability that $X$ is chosen, $P_i$ for all $i \in [m]$ denoting the price paid by bidder $i$, and $V_\obj$, denoting the value of $\obj$. We say that a vector denoting a randomized allocation, price vector, and value of $\obj$ is feasible if it corresponds to an actual distribution over elements of $\mathcal{F}$, charges each bidder a price in $[0,1]$, and correctly computes (or underestimates) the value of $\obj$. Specifically, this corresponds to the following constraints:

\begin{enumerate}
\item $Q(X) \geq 0$ for all $X \in \mathcal{F}$. This guarantees that all probabilities are non-negative. There are $|\mathcal{F}|$ of these constraints.
\item $\sum_{X \in \mathcal{F}} Q(X) = 1$. This guarantees that all probabilities sum to one. There is one of these constraints.
\item $0 \leq P_i \leq 1$ for all $i \in [m]$. This guarantees that all prices are bounded. There are $m$ of these constraints.
\item $0 \leq V_\obj \leq \obj(Q,P)$. This guarantees that the value of $\obj$ is correctly underestimated. There is one of these constraints.
\end{enumerate}


Observe now that all of these constraints are linear except for 4). In fact, because we know that $\obj$ is concave (and therefore $\obj(Q,P) = \min_{\vec{c} \in S} \{\vec{c} \cdot (Q,P)\}$ \mattnote{for some (possibly infinite) set $S$ of linear functions}), we can rewrite 4) as the intersection of several linear constraints:

$$0 \leq V_\obj \leq \vec{c} \cdot (Q,P)\text{ for all $\vec{c} \in S$. There are $|S|$ of these constraints.}$$

We have now explicitly written the halfspaces whose intersection defines $F'$, and therefore all extreme points of $F'$ lie in the intersection of $|\mathcal{F}| + m + 1$ of the corresponding hyperplanes. Observe now that there are two kinds of hyperplanes: those of the form $Q(X) = 0$ (from 1) above), and everything else. The total number of constraints not from 1) above is just \mattnote{$2m + 2 + |S|$}, so the total number of such corresponding hyperplanes is certainly polynomial in $\sum_i |T_i|, |S|$. With this observation, we can bound the bit complexity of extreme points of $F'$. 

Any extreme point is the solution to a system of equations that immediately forces all except for $\poly(\sum_i |T_i|, |S|)$ variables to be $0$, and whose remaining $\poly(\sum_i |T_i|,|S|)$ variables solve a system of $\poly(\sum_i |T_i|,|S|)$ equations whose coefficients all have bit complexity $b$. Therefore, we may conclude\footnote{Recall that we defined the bit complexity $b$ of $\obj$ in a way so that $b \geq |S|$ always.} that all extreme points of $F'$ have coordinates with bit complexity $\poly(b,\sum_i |T_i|)$. Now, it is easy to see that every virtual objective is a linear function over $F'$, and therefore all virtual objective maximizers are extreme points of $F'$. In conclusion, we have shown that on every profile, all virtual objective maximizers correspond to randomized allocations that only assign probabilities and prices of bit complexity $\poly(b,\sum_i |T_i|)$, \mattnote{and whose objective value also has bit complexity $\poly(b,\sum_i |T_i|)$.}

The last step is now to understand how this relates to the bit complexity of implicit forms of mechanisms. We have just shown above that every virtual objective maximizer assigns probabilities and prices of bit complexity $\poly(b,\sum_i |T_i|)$, \mattnote{and obtains an objective value of bit complexity $\poly(b,\sum_i |T_i|)$} on every profile. The expected price paid by bidder $i$ when reporting type $t_i$ is clearly just a convex combination of the prices charged on each profile, where each coefficient corresponds to the probability that the profile is chosen. Each such probability has bit complexity \mattnote{$\poly(b,\sum_i |T_i|)$ (actually it is much smaller because $\mathcal{D}'$ is a uniform distribution)}, and \mattnote{because there are only $\poly(b,\sum_i |T_i|)$ profiles, }the expected price paid by bidder $i$ when reporting $t_i$ in any virtual objective maximizer \mattnote{is the sum of $\poly(b,\sum_i |T_i|)$ terms of bit complexity $\poly(b,\sum_i |T_i|)$. Therefore, the bit complexity of the expected price paid by bidder $i$ when reporting type $t_i$ is $\poly(b,\sum_i |T_i|)$}. Similarly, the expected value of bidder $i$ of type $t'_i$ when reporting type $t_i$ is just a convex combination of the values of $t'_i$ for the randomized allocations awarded on each profile where bidder $i$'s type is $t_i$. As the value of $t'_i$ for any allocation in $\mathcal{F}$ has bit complexity $b$, and for every profile the probabilities assigned to each allocation has bit complexity $\poly(b,\sum_i |T_i|)$, each of these values will also have bit complexity $\poly(b,\sum_i |T_i|)$. Again, the coefficients in the convex combination correspond to the probability that the profile is chosen, which has bit complexity \mattnote{$\poly(b,\sum_i |T_i|)$}, so the expected value of bidder $i$ with type $t'_i$ for reporting $t_i$ in any virtual objective maximizer is also $\poly(b,\sum_i |T_i|)$. Finally, the value of $\pi_\obj$ is just the expected value of $V_\obj$ over all profiles, which is again a convex combination using coefficients of bit complexity \mattnote{$\poly(b,\sum_i |T_i|)$}. So $\pi_\obj$ also has bit complexity $\poly(b,\sum_i |T_i|)$ in  any virtual objective maximizer. We have now shown that every component of the implicit form of any virtual objective maximizer has bit complexity $\poly(b,\sum_i |T_i|)$, and therefore all extreme points of \mattnote{$F(\mathcal{F},\mathcal{D}',\obj)$} have bit complexity $\poly(b,\sum_i |T_i|)$ as desired.

\end{prevproof}

\begin{prevproof}{Theorem}{thm:objective}
\mattnote{The only part of the theorem that is not an obvious consequence of previous lemmas and corollaries} is that the final mechanism output matches the desired implicit form. Recall that each $\vec{\pi}^I_j = \mathcal{A}(\vec{w}_j)$ has a corresponding mechanism $M_j$ with $\vec{\pi}(M_j) = \vec{\pi}_j$, $\vec{P}(M_j) = \vec{P}_j$, and $\pi_\obj(M_j) \geq (\pi_\obj)_j$. The mechanism output is exactly $M = \sum_j c_j M_j$. So the implicit form of this mechanism satisfies $\vec{\pi}(M) = \vec{\pi}$, $\vec{P}(M) = \vec{P}$, and $\pi_\obj(M) \geq \sum_j c_j \pi_\obj(M_j) \geq \pi_\obj$ because $\obj$ is concave. Therefore, the mechanism output is at least as good as the implicit form promises (and perhaps better).

\end{prevproof}

\section{Omitted Proofs from Section~\ref{sec:LPrevenue}}\label{app:LPrevenue}
\begin{prevproof}{Lemma}{lem:convexpolytope}
Every randomized mechanism is a distribution over deterministic mechanisms. Furthermore, the implicit form of a distribution over deterministic mechanisms is just the corresponding convex combination of their implicit forms. So if we let $S$ denote the set of implicit forms of deterministic mechanisms, $F(\mathcal{F},\mathcal{D})$ is contained in the convex hull of $S$. Furthermore, any implicit form that is a convex combination of points in $S$ can be implemented by a mechanism that is the corresponding distribution over the corresponding mechanisms. Therefore, $F(\mathcal{F},\mathcal{D})$ is exactly the convex hull of $S$.
\end{prevproof}

\begin{prevproof}{Observation}{obs:LPsolvesMDMDP}
It is clear that any point output can be interpreted as an implicit form. It is also clear that the feasible region is exactly the set of feasible, BIC, IR implicit forms. It is also clear that the objective function evaluated at a point is just the expected revenue of the corresponding implicit form. Finally, it is clear that every feasible, BIC, IR mechanism has a feasible, BIC, IR implicit form. So the optimal achievable value of the LP's objective function within the feasible region is exactly the expected revenue of the optimal feasible, BIC, IR mechanism.
\end{prevproof}

\begin{prevproof}{Corollary}{cor:LP}
It is clear that the variable count is correct, and that the BIC and IR constraints are linear. By Lemma~\ref{lem:convexpolytope}, $F(\mathcal{F},\mathcal{D})$ is a convex polytope, and therefore the entire feasible region is a convex polytope. It is also clear that the objective function is linear. \mattnote{Finally, it is clear that if $b$ upper bounds the bit complexity of $Pr[t_i]$ and $t_i(X)$ for all $i,t_i \in T_i, X \in \mathcal{F}$, then all extreme points of $F(\mathcal{F},\mathcal{D})$ have bit complexity $\poly(b,\sum_i |T_i|)$. }

That an $\alpha$-approximate solution can be found using $WSO$ is a consequence of Corollary~\ref{cor:CaiDW1}. Denote by $\vec{\pi}_1$ the optimal solution to the LP in Figure~\ref{fig:LPMDMD}, $\vec{\pi}_2$ the optimal solution after replacing $F(\mathcal{F},\mathcal{D})$ with $\alpha F(\mathcal{F},\mathcal{D})$ and $\vec{\pi}_3$ the optimal solution after replacing $F(\mathcal{F},\mathcal{D})$ with $WSO$ for $\alpha F(\mathcal{F},\mathcal{D})$. Denote by $c_1,c_2$, and $c_3$ the expected revenue obtained by each. It is easy to see that $\alpha\cdot\vec{\pi}_1$ is a feasible solution to the LP using $\alpha F(\mathcal{F},\mathcal{D})$: $\alpha\cdot\vec{\pi}_1$ is in $\alpha F(\mathcal{F},\mathcal{D})$ by definition, and the BIC and IR constraints are invariant to scaling. Furthermore, it is clear that the expected revenue of $\alpha\cdot\vec{\pi}_1$ is exactly $\alpha\cdot c_1$. We may therefore conclude that $c_2 \geq \alpha\cdot c_1$. As $WSO$ is a weird separation oracle for $\alpha F(\mathcal{F},\mathcal{D})$, Corollary~\ref{cor:CaiDW1} guarantees that $c_3 \geq c_2 \geq \alpha\cdot c_1$, and that $\vec{\pi}_3$ can be found using the same parameters as a valid separation oracle for $\alpha F(\mathcal{F},\mathcal{D})$. So using Theorem~\ref{thm:ellipsoid} and replacing the running time of $SO$ with the running time of $WSO$, we obtain the desired running time. Finally, by Observation~\ref{obs:LPsolvesMDMDP}, $\vec{\pi}_3$ corresponds to the implicit form of an $\alpha$-approximate solution to MDMDP.
\end{prevproof}

\section{Omitted Details from Section~\ref{sec:MDMDPtoSADP}}\label{app:MDMDPtoSADP}

\begin{prevproof}{Proposition}{prop:virtualwelfare}
\begin{align*}
\vec{w} \cdot \vec{x} =  \sum_{i\in[m]}\sum_{t_i\in T_{i},t'_i \in T_{i}} \pi_i(t_i,t'_i) w_i(t_i,t'_i) = \sum_{i\in[m]}\sum_{t_i\in T_{i},t'_i \in T_{i}} \pi_i(t_i,t'_i) (w_i(t_i,t'_i)/\Pr[t'_i]) \Pr[t'_i]
\end{align*}

Now, consider any mechanism $M$ with interim form $\vec{\pi}$ and let $X_i(t'_i)$ denote the (randomized) interim allocation that bidder $i$ receives when reporting type $t'_i$ to $M$. Observe that $\pi_i(t_i,t'_i)$ is just $t_i(X_i(t'_i))$, and does not depend on the specific choice of $M$ or $X_i(t'_i)$. So after making this substitution, we have:

\begin{equation}\label{eq:dotproduct}
\vec{w} \cdot \vec{\pi} = \sum_{i\in[m],t'_i\in T_{i}} \Pr[t'_i] \sum_{t_i\in T_{i}} t_i(X_i(t'_i)) (w_i(t_i,t'_i)/\Pr[t'_i])
\end{equation}

Now, $\sum_{t_i\in T_{i}} (w_i(t_i,t'_i)/\Pr[t'_i])t_i(\cdot)$ is a function that takes as input an interim allocation $X_i(t'_i)$ and outputs a value $\sum_{t_i\in T_{i}} (w_i(t_i,t'_i)/\Pr[t'_i])t_i(X_i(t'_i))$. We may think of $\sum_{t_i\in T_{i}} (w_i(t_i,t'_i)/\Pr[t'_i])t_i$ as the \emph{virtual type} of bidder $i$ with type $t'_i$, and $\sum_{t_i\in T_{i}} (w_i(t_i,t'_i)/\Pr[t'_i])t_i(X_i(t'_i))$ as the virtual value of that virtual type for the allocation $X_i(t'_i)$. Under this interpretation, $\vec{w} \cdot \vec{x}$ is summing over all bidders $i$ and all types $t'_i$, the expected virtual value obtained by $t'_i$, which is exactly the expected virtual welfare of a mechanism with interim form $\vec{\pi}$ computed with respect to the desired virtual types.
\end{prevproof}

\begin{prevproof}{Corollary}{cor:obv}
By Proposition~\ref{prop:virtualwelfare}, we know that $\vec{x} \cdot \vec{w}$ is exactly the virtual welfare of a mechanism with implicit form $\vec{x}$. A mechanism that obtains a $\alpha$-fraction of the optimal virtual welfare on every profile clearly obtains a $\alpha$-fraction of the optimal expected virtual welfare in expectation.
\end{prevproof}

\begin{prevproof}{Proposition}{prop:SADPvirtualwelfare}
First, it is clear that $(f_1,f_2)$ is a valid input to SADP($\mathcal{F}$,$\mathcal{V}$), as $f_1,f_2 \in \mathcal{V}^*$. Furthermore, because only two functions are input to this instance of SADP, any $X$ that is an $\alpha$-approximation necessarily satisfies:

\begin{align*}
& \left(\sum_i \sum_{t_i| C_i(t_i) > 0} C_i(t_i) t_i(X) - \sum_i \sum_{t_i | C_i(t_i) < 0} - C_i(t_i) t_i(X)\right)\\ 
\geq &\alpha\cdot\max_{X' \in \mathcal{F}} \left\{\left(\sum_i \sum_{t_i| C_i(t_i) > 0} C_i(t_i) t_i(X') - \sum_i \sum_{t_i | C_i(t_i) < 0} - C_i(t_i) t_i(X')\right)\right\}.
\end{align*}
It is easy to see that the above equation can be rearranged to the desired form.
\end{prevproof}

As noted in Section~\ref{sec:MDMDPtoSADP}, at this point, we are almost ready to give our algorithm. We would like to just let $\mathcal{A}$ be the algorithm that takes as input a direction $\vec{w}$ and computes the implicit form prescribed by Corollary~\ref{cor:important}. Corollary~\ref{cor:important} shows that this algorithm satisfies the hypotheses of Theorem~\ref{thm:CaiDW1}, so we would get a weird separation oracle for $\alpha F(\mathcal{F},\mathcal{D})$. We noted in Section~\ref{sec:MDMDPtoSADP} that this could not be done trivially in an efficient manner, and elaborate below on the required techniques to overcome this developed in~\cite{CaiDW12b,CaiDW13}.

First, computing the implicit form prescribed in Corollary~\ref{cor:important} would require enumerating every profile in the support of $\mathcal{D}$ (and there are exponentially many of them). The obvious approach to cope with this would be to just approximate the implicit form by sampling polynomially many profiles from $\mathcal{D}$. Unfortunately again, this approach does not work as the inconsistent error due to sampling causes the geometry of LP solvers to fail. Luckily, however, a similar problem arises and is solved in~\cite{CaiDW12b}. Simply put, the approach is to take polynomially many samples from $\mathcal{D}$ to form a multi-set $S$, then define a new distribution $\mathcal{D}'$ that samples a profile uniformly at random from $S$. $F(\mathcal{F},\mathcal{D}')$ is still well-defined, and one might expect that it should look a lot like $F(\mathcal{F},\mathcal{D})$ if enough samples are taken. {This is indeed the case, as is and shown formally in Section~5.2 of~\cite{CaiDW12b} (to which we refer the reader for the complete details of the approach).}

Second, even after enumerating every profile, we would still have to enumerate the randomness used in the (possibly randomized) SADP solutions output by $G$. {We solve this problem by first observing that we can run polynomially many independent trials of $G$ to obtain an algorithm $G'$ that obtains an $(\alpha-\epsilon/\poly(m))$-approximation with very high probability. Then, we can fix the randomness used ahead of time and take a union bound so that with very high probability, $G'$ obtains an $(\alpha - \epsilon/\poly(m))$-approximation every time it's queried during the algorithm. As the randomness is fixed ahead of time, we may just treat $G'$ as a deterministic algorithm. This process is discussed formally in Appendix~G.2 of~\cite{CaiDW13} (to which we refer the reader for the complete details of the approach).}

After these two modifications, we now have an algorithm that can find an approximately optimal implicit form. That's great, but how can we turn this into an actual mechanism? Luckily, Theorem~\ref{thm:CaiDW1} answers this for us. Whatever implicit form $\vec{\pi}$ is output was deemed feasible by the weird separation oracle. Theorem~\ref{thm:CaiDW1} guarantees then that the weird separation oracle found explicit directions, $\vec{w}_1,\ldots,\vec{w}_l$ such that $\vec{\pi}$ lies within the convex hull of the $\mathcal{A}(\vec{w}_j)$. But recall that $\mathcal{A}(\vec{w}_j)$ is just the implicit form of the mechanism that runs $G'$ on every profile, using as input the virtual types corresponding to $\vec{w}$. This exactly says that, if $\vec{\pi} = \sum_j c_j \mathcal{A}(\vec{w}_j)$, the mechanism that first samples a virtual transformation $\vec{w}$, choosing $\vec{w}_j$ with probability $c_j$, and then runs $G'$ using these virtual types has implicit form exactly $\vec{\pi}$. With this discussion in mind, we are now ready to state formally our reduction from MDMDP to SADP:
\begin{algorithm}[ht]
        \caption{Algorithm to solve MDMDP($\mathcal{F}$,$\mathcal{V}$):}
      
    \begin{algorithmic}[1]\label{alg:solveMDMDP}
        \STATE Input: \mattnote{type spaces $T_i$ for $i \in [m]$,} distributions $\mathcal{D}_i$ over $T_{i}$, an error $\epsilon$, and black-box access to $G$, a (possibly randomized) $\alpha$-approximation algorithm for SADP($\mathcal{F}$,$\mathcal{V}$).
        \STATE {Obtain $\mathcal{D}'$ as the uniform distribution over a multi-set of profiles sampled from $\mathcal{D}$ in the manner described in Section~5.2 of~\cite{CaiDW12b}.}
        \STATE Modify $G$ to obtain $G'$ as prescribed in {Appendix~G.2 of~\cite{CaiDW13}}. $G'$ is a deterministic algorithm.
        \STATE Define $\mathcal{A}$ to, on input $\vec{w} \in [-1,1]^{\sum_i |T_i|^2}$, compute the implicit form prescribed in Corollary~\ref{cor:important} using $G'$ and $\mathcal{D}'$. 
         \STATE Solve the linear program in Figure~\ref{fig:LPMDMD} using the weird separation oracle obtained from $\mathcal{A}$ using Theorem~\ref{thm:CaiDW1}. Call the output $(\vec{\pi},\vec{P})$, and let $\vec{w}_1,\ldots,\vec{w}_l$ be the directions guaranteed by Theorem~\ref{thm:CaiDW1} to be output by the weird separation oracle.
         \STATE Solve the linear system $\vec{\pi} = \sum_{j\in[l] }c_j \mathcal{A}(\vec{w}_j)$ to write $\vec{\pi}$ as a convex combination of $\mathcal{A}(\vec{w}_j)$.
         \STATE Output the following mechanism: First, randomly select a $\vec{w} \in [-1,1]^{\sum_i |T_i|^2}$ by choosing $\vec{w}_j$ with probability $c_j$. Then, choose the allocation output by $G'$ on the SADP input corresponding to $\vec{w}$.
             \end{algorithmic}
\end{algorithm}
\section{Omitted Proofs from Section~\ref{sec:maxmin}}\label{app:maxmin}

\begin{prevproof}{Observation}{obs:concave}
Consider two allocations $X_1$ and $X_2$. Then:
\begin{align*}
FMMF(\vec{t},cX_1 + (1-c)X_2) &= \min_i \{t_i(cX_1+(1-c)X_2)\} \\
&\geq c \cdot \min_i \{t_i(X_1)\} + (1-c) \cdot \min_i \{t_i(X_2)\}\\
&=c\cdot FMMF(\vec{t},X_1) + (1-c)\cdot FMMF(\vec{t},X_2)
\end{align*}
\end{prevproof}

\begin{prevproof}{Observation}{obs:LPsolvesmaxmin}
It is easy to see that the objective function is less than $\obj'(\vec{g},X,f,f')$, because $O$ is less than $FMMF(\vec{g},X)$, and the rest of $\obj'(\vec{g},X,f,f')$ is computed correctly in the objective function. It is also easy to see that the maximum value of the LP is exactly $\max_{X \in \Delta(\mathcal{F})}\{\obj'(\vec{g},X,f,f')\}$. Putting these together proves the observation.
\end{prevproof}

\begin{prevproof}{Corollary}{cor:LPmaxmin}
It is clear that the variable count is correct, and that the feasible region is a convex polytope. It is also clear that the objective function is linear. The remainder of the proof is nearly identical to that of Corollary~\ref{cor:LP} but is included below for completeness.

That an $\alpha$-approximate solution can be found using $WSO$ is a consequence of Corollary~\ref{cor:CaiDW1}. Denote by $\vec{X}_1$ the optimal solution to the LP in Figure~\ref{fig:LPmaxmin}, $\vec{X}_2$ the optimal solution after replacing $\Delta(\mathcal{F})$ with $\alpha \Delta(\mathcal{F})$ and $\vec{X}_3$ the optimal solution after replacing $\Delta(\mathcal{F})$ with $WSO$ for $\alpha \Delta(\mathcal{F})$. Denote by $c_1,c_2$, and $c_3$ the value of the objective function at each. It is easy to see that $\alpha\cdot\vec{X}_1$ is a feasible solution to the LP using $\alpha \Delta(\mathcal{F})$: $\alpha\cdot\vec{X}_1$ is in $\alpha \Delta(\mathcal{F})$ by definition, and the remaining constraints are invariant to scaling. Furthermore, it is clear that the value of the objective function evaluated at $\alpha\cdot\vec{X}_1$ is exactly $\alpha\cdot c_1$. We may therefore conclude that $c_2 \geq \alpha\cdot c_1$. As $WSO$ is a weird separation oracle for $\alpha \Delta(\mathcal{F})$, Corollary~\ref{cor:CaiDW1} guarantees that $c_3 \geq c_2 \geq \alpha\cdot c_1$, and that $\vec{X}_3$ can be found using the same parameters as a valid separation orcale for $\alpha \Delta(\mathcal{F})$. So using Theorem~\ref{thm:ellipsoid} and replacing the running time of $SO$ with the running time of $WSO$, we obtain the desired running time. Finally, by Observation~\ref{obs:LPsolvesmaxmin}, $\vec{X}_3$ corresponds to an $\alpha$-approximate solution to maximizing $\obj'(\vec{g},X,f,f')$ over $\Delta(\mathcal{F})$.
\end{prevproof}

\begin{prevproof}{Observation}{obs:obvious}
We can write the expected weight of a set sampled from $X$ $\sum_{i,j} \Pr[(i,j) \in X]w_{ij}$, which is exactly $\vec{w} \cdot \vec{X}$.
\end{prevproof}

\bibliographystyle{abbrv}
\bibliography{costasbibold}

\end{document}